\documentclass[11pt]{article}

\usepackage{microtype}
\usepackage[dvipsnames,svgnames,x11names]{xcolor}
\usepackage{amsmath, amsthm, thm-restate}
\usepackage{xspace}
\usepackage{float}
\usepackage[T1]{fontenc}
\usepackage{mdframed}
\usepackage{mathtools}
\usepackage{enumitem}
\usepackage{titlesec}
\usepackage[font=small]{caption}

\titlespacing*{\paragraph}{0pt}{1.5ex plus 0.5ex minus 0.2ex}{1em}

\usepackage{authblk}

\usepackage{fullpage}
\usepackage{thmtools}
\usepackage[colorlinks,hypertexnames = false]{hyperref}
\hypersetup{colorlinks={true},linkcolor={blue},citecolor=magenta}

\usepackage{tikz}
\usetikzlibrary{calc}
\usepackage{pgfcalendar}

\usepackage{amssymb,amsfonts,amsmath,amsthm,amscd,dsfont,mathrsfs}
\usepackage{complexity}
\usepackage{multirow}
\usepackage{mathpazo}
\usepackage{braket}
\usepackage{comment}
\usepackage{quantikz}
\usepackage[
backend=biber,
style=alphabetic,
maxalphanames=4,
sorting=ynt,
maxnames=99
]{biblatex}
\usepackage{url}
\usepackage{bbm}
\usepackage[capitalize,nameinlink]{cleveref}
\usepackage{hyperref} 
\usepackage[margin=1in]{geometry}
\hypersetup{breaklinks=true}

\theoremstyle{plain}
\newtheorem{theorem}{Theorem}[section]
\newtheorem{lemma}[theorem]{Lemma}
\newtheorem*{claim*}{Claim}
\newtheorem{corollary}[theorem]{Corollary}
\newtheorem{definition}[theorem]{Definition}
\newtheorem{proposition}[theorem]{Proposition}
\newtheorem{remark}[theorem]{Remark}

\newtheorem*{theorem*}{Theorem}
\newtheorem*{definition*}{Definition}

\let\originalleft\left
\let\originalright\right
\renewcommand{\left}{\mathopen{}\mathclose\bgroup\originalleft}
\renewcommand{\right}{\aftergroup\egroup\originalright}

\newtheoremstyle{boxed}{-\topsep}{}{\itshape}{}{\bfseries}{.}{.5em}{}
\theoremstyle{boxed}

\newmdtheoremenv[backgroundcolor=gray!10,
                 linewidth=0pt,
                 innerleftmargin=4pt,
                 innerrightmargin=4pt,
                 innertopmargin=4pt,
                 innerbottommargin=4pt,
            splitbottomskip=4pt,
            nobreak=true]{problem}[prob]{Problem}

\renewcommand{\set}[1]{\left\{ #1 \right\}}

\mathchardef\mhyphen="2D

\newcommand{\Q}{\mathsf{Q}}

\renewcommand{\C}{\mathbb{C}}

\newcommand{\bit}{\{0,1\}}

\title{Optimal Quantum Algorithms for Ordered Search}
\author{Joseph Carolan and Andrew M.\ Childs}

\affil{Department of Computer Science, Institute for Advanced Computer Studies, and \\
Joint Center for Quantum Information and Computer Science,
University of Maryland}

\date{\vspace{-2em}}

\begin{document}

\maketitle

\begin{abstract}
    Ordered search is the problem of locating a target element in a sorted list of size $n$ using comparison queries. Classically, binary search requires $\lceil \log_2 n\rceil$ queries, which is optimal. Quantum algorithms offer a constant-factor speedup, but the precise constant has been a longstanding open question.
    We close this gap by exhibiting two new quantum algorithms for ordered search, each using the optimal $\frac{1}{\pi}\ln n+o(\log n)$ queries. The first, discovered by Claude Fable 5, is a simple zero-error algorithm derived from a continuum relaxation of the problem. The second, discovered by GPT-5.6-Sol (informed by Claude's zero-error algorithm), is an exact algorithm based on an analytic solution of the polynomial program of Farhi, Goldstone, Gutmann, and Sipser.

\end{abstract}

\section{Introduction}
\label{sec:intro}

Ordered search is the problem of locating a target element in a sorted list using comparison queries.  Classically, this fundamental problem is completely understood: binary
search finds the target in a length $n$ list using $\lceil\log_2n\rceil$ comparisons, and the information-theoretic bound of $\log_2n$ shows that this is optimal in the leading constant.

The quantum complexity of ordered search has resisted a similarly
complete answer for over twenty-five years.  Quantum algorithms can beat
binary search by at most a constant factor, but the value of the optimal
constant has remained open. While a constant-factor quantum speedup is unlikely to be practically relevant, the problem is fundamental and provides a natural testbed for developing new quantum algorithms beyond the scope of prior techniques.

On the algorithmic side, Farhi, Goldstone, Gutmann, and
Sipser introduced the class of \emph{translation-invariant} algorithms,
characterized exact algorithms in this class by a feasibility program
over nonnegative Laurent polynomials, and numerically found a $3$-query
algorithm searching $52$ elements, giving a recursive algorithm searching a length-$n$ list with
$3\log_{52}n\approx0.526\log_2n$ queries~\cite{FarhiGoldstoneGutmannSipser1999}. H\o yer, Neerbek, and Shi gave an exact
algorithm based on a different method using $\log_3n\approx0.631\log_2n$
queries~\cite{HoyerNeerbekShi2002}, though it performs worse.
Subsequent improvements followed the template of numerically
finding a small translation-invariant algorithm and recursing: Brookes, Jacokes, and Landahl
found a $4$-query algorithm for $550$ elements via gradient
descent~\cite{BrookesJacokesLandahl2004}; Childs, Landahl, and Parrilo
reformulated the feasibility program as a semidefinite program and
searched $605$ elements with $4$ queries, giving
$4\log_{605}n\approx0.433\log_2n$~\cite{ChildsLandahlParrilo2007}; Carolan, Childs, Kovacs-Deak, and Schaeffer recently used a linear
programming relaxation to show that the largest list searchable with
$5$ queries has exactly $7265$ elements, giving  $5\log_{7265}n\approx0.390\log_2n$~\cite{CarolanChildsKovacsDeakSchaeffer2025}; and Wu et al.\ extended that approach using GPU acceleration to give the best known exact upper bound of $6\log_{90000} n \approx 0.365 \log_2 n$~\cite{wu2026matrixfree}.
Relaxing exactness, Ben-Or and Hassidim gave a zero-error algorithm
using roughly $0.323\log_2n$ queries in
expectation~\cite{BenOrHassidim2007}.

On the lower-bound side, Buhrman and de Wolf proved
$\Omega(\sqrt{\log n}/\log\log n)$~\cite{BuhrmanDeWolf1999}, improved to
$\Omega(\log n/\log\log n)$ by Farhi, Goldstone, Gutmann, and
Sipser~\cite{FarhiGoldstoneGutmannSipser1998} and to
$\frac1{12}\log_2n$ by Ambainis~\cite{Ambainis1999}.  Using the
adversary method~\cite{Ambainis2002}, H\o yer, Neerbek, and Shi proved
the best known lower bound of
\[
    \frac1\pi(\ln n-1)\approx0.221\log_2n,
\]
which holds for exact algorithms and degrades multiplicatively with
error~\cite{HoyerNeerbekShi2002}.  Childs and Lee later showed that this
is optimal for the adversary method: the adversary value of ordered
search is $\frac2\pi\ln n\pm O(1)$, even for the negative adversary, so no adversary-type argument can improve the constant~\cite{ChildsLee2008}.  They speculated that the
lower bound is tight and that progress would come from algorithms. This history is depicted in \Cref{fig:history}.

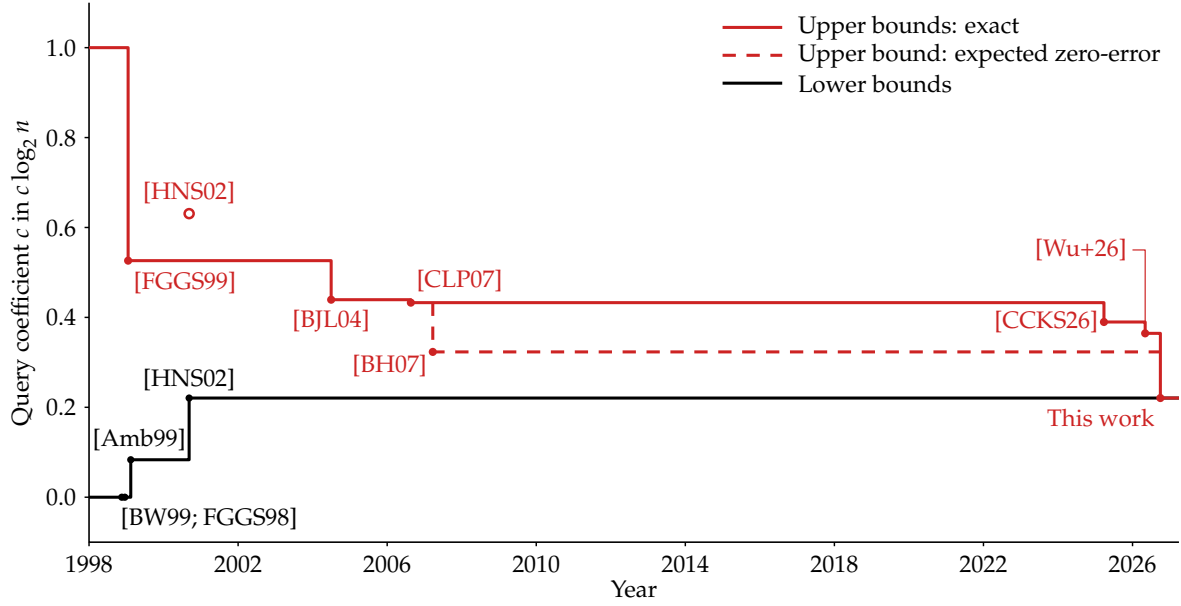
\begin{figure}
    \centering
\definecolor{historyUpper}{RGB}{205,35,35}%
\hypersetup{citecolor=historyUpper}%
\begin{tikzpicture}[
  scale=0.85,
  x=0.58cm, y=7cm,
  font=\fontsize{9}{11}\selectfont,
  line cap=butt, line join=round,
  history upper/.style={draw=historyUpper,line width=1.15pt},
  history lower/.style={draw=black,line width=1.15pt},
  history expected/.style={history upper,dash pattern=on 5pt off 4pt},
  history upper label/.style={text=historyUpper,inner sep=2pt},
  history lower label/.style={text=black,inner sep=2pt},
  history leader/.style={line width=0.45pt}
]

\def\HistoryOriginDate{1998-01-01}
\def\HistoryEndDate{2027-04-01}

\def\HistoryDateBdW{1998-11-18}
\def\HistoryDateFGGSLower{1998-12-18}
\def\HistoryDateFGGSUpper{1999-01-19}
\def\HistoryDateAmbainis{1999-02-14}
\def\HistoryDateHNS{2000-09-08}
\def\HistoryDateBJL{2004-07-01}
\def\HistoryDateCLP{2006-08-21}
\def\HistoryDateBH{2007-03-24}
\def\HistoryDateCCKS{2025-03-27}
\def\HistoryDateWu{2026-05-05}
\def\HistoryDateThisWork{2026-10-01}

\pgfmathsetmacro{\HistoryCClassical}{1}
\pgfmathsetmacro{\HistoryCSublog}{0}
\pgfmathsetmacro{\HistoryCAmbainis}{1/12}
\pgfmathsetmacro{\HistoryCOptimal}{ln(2)/pi}
\pgfmathsetmacro{\HistoryCFGGS}{3*ln(2)/ln(52)}
\pgfmathsetmacro{\HistoryCHNSUpper}{ln(2)/ln(3)}
\pgfmathsetmacro{\HistoryCBJL}{4*ln(2)/ln(550)}
\pgfmathsetmacro{\HistoryCCLP}{4*ln(2)/ln(605)}
\pgfmathsetmacro{\HistoryCBH}{6/18.5625}
\pgfmathsetmacro{\HistoryCCCKS}{5*ln(2)/ln(7265)}
\pgfmathsetmacro{\HistoryCWu}{6*ln(2)/(ln(9)+4*ln(10))}

\newcount\HistoryOriginJulian
\newcount\HistoryJulian
\pgfcalendardatetojulian{\HistoryOriginDate}{\HistoryOriginJulian}
\newcommand{\HistoryDateToX}[2]{%
  \pgfcalendardatetojulian{#2}{\HistoryJulian}%
  \edef\HistoryDayOffset{\number\numexpr\HistoryJulian-\HistoryOriginJulian\relax}%
  \pgfmathsetmacro{#1}{\HistoryDayOffset/365.2425}%
}
\HistoryDateToX{\HistoryXBdW}{\HistoryDateBdW}
\HistoryDateToX{\HistoryXFGGSLower}{\HistoryDateFGGSLower}
\HistoryDateToX{\HistoryXFGGSUpper}{\HistoryDateFGGSUpper}
\HistoryDateToX{\HistoryXAmbainis}{\HistoryDateAmbainis}
\HistoryDateToX{\HistoryXHNS}{\HistoryDateHNS}
\HistoryDateToX{\HistoryXBJL}{\HistoryDateBJL}
\HistoryDateToX{\HistoryXCLP}{\HistoryDateCLP}
\HistoryDateToX{\HistoryXBH}{\HistoryDateBH}
\HistoryDateToX{\HistoryXCCKS}{\HistoryDateCCKS}
\HistoryDateToX{\HistoryXWu}{\HistoryDateWu}
\HistoryDateToX{\HistoryXThisWork}{\HistoryDateThisWork}
\HistoryDateToX{\HistoryXEnd}{\HistoryEndDate}

\coordinate (BdW)       at (\HistoryXBdW,\HistoryCSublog);
\coordinate (FGGSLower) at (\HistoryXFGGSLower,\HistoryCSublog);
\coordinate (FGGSUpper) at (\HistoryXFGGSUpper,\HistoryCFGGS);
\coordinate (Ambainis)  at (\HistoryXAmbainis,\HistoryCAmbainis);
\coordinate (HNSLower)  at (\HistoryXHNS,\HistoryCOptimal);
\coordinate (HNSUpper)  at (\HistoryXHNS,\HistoryCHNSUpper);
\coordinate (BJL)       at (\HistoryXBJL,\HistoryCBJL);
\coordinate (CLP)       at (\HistoryXCLP,\HistoryCCLP);
\coordinate (BH)        at (\HistoryXBH,\HistoryCBH);
\coordinate (CCKS)      at (\HistoryXCCKS,\HistoryCCCKS);
\coordinate (Wu)        at (\HistoryXWu,\HistoryCWu);
\coordinate (ThisWork)  at (\HistoryXThisWork,\HistoryCOptimal);

\draw[line width=0.6pt] (0,1.10) -- (0,-0.10) -- (\HistoryXEnd,-0.10);
\foreach \year in {1998,2002,2006,2010,2014,2018,2022,2026}{%
  \HistoryDateToX{\HistoryXTick}{\year-01-01}%
  \draw[line width=0.55pt] (\HistoryXTick,-0.10) -- ++(0,-0.009);
  \node[anchor=north,inner sep=3pt] at (\HistoryXTick,-0.11) {\year};
}
\foreach \y/\labeltext in {0/0.0,0.2/0.2,0.4/0.4,0.6/0.6,0.8/0.8,1/1.0}{%
  \draw[line width=0.55pt] (0,\y) -- ++(-0.12,0);
  \node[anchor=east,inner sep=3pt] at (-0.12,\y) {\labeltext};
}
\node at ({\HistoryXEnd/2},-0.205) {Year};
\node[rotate=90] at (-1.75,0.50)
  {Query coefficient $c$ in $c\log_2 n$};

\draw[history lower]
  (0,0) -- (\HistoryXAmbainis,0) -- (Ambainis)
  -- (\HistoryXHNS,\HistoryCAmbainis) -- (HNSLower)
  -- (\HistoryXEnd,\HistoryCOptimal);

\draw[history upper]
  (0,\HistoryCClassical) -- (\HistoryXFGGSUpper,\HistoryCClassical)
  -- (FGGSUpper) -- (\HistoryXBJL,\HistoryCFGGS) -- (BJL)
  -- (\HistoryXCLP,\HistoryCBJL) -- (CLP)
  -- (\HistoryXCCKS,\HistoryCCLP) -- (CCKS)
  -- (\HistoryXWu,\HistoryCCCKS) -- (Wu)
  -- (\HistoryXThisWork,\HistoryCWu) -- (ThisWork)
  -- (\HistoryXEnd,\HistoryCOptimal);

\draw[history expected] (\HistoryXBH,\HistoryCCLP) -- (BH)
  -- (\HistoryXThisWork,\HistoryCBH);

\foreach \p in {BdW,FGGSLower,Ambainis,HNSLower}
  \fill[black] (\p) circle[radius=1.6pt];
\foreach \p in {FGGSUpper,BJL,CLP,BH,CCKS,Wu,ThisWork}
  \fill[historyUpper] (\p) circle[radius=1.8pt];
\draw[historyUpper,line width=0.95pt,fill=white]
  (HNSUpper) circle[radius=2pt];

\node[history upper label,anchor=north west]
  at (\HistoryXFGGSUpper,0.515) {\cite{FarhiGoldstoneGutmannSipser1999}};
\node[history upper label,anchor=south]
  at (\HistoryXHNS,0.64) {\cite{HoyerNeerbekShi2002}};
\node[history upper label,anchor=north]
  at (\HistoryXBJL,0.43) {\cite{BrookesJacokesLandahl2004}};
\node[history upper label,anchor=south west]
  at (\HistoryXCLP,0.445) {\cite{ChildsLandahlParrilo2007}};
\node[history upper label,anchor=north east]
  at (\HistoryXBH,0.33) {\cite{BenOrHassidim2007}};

\node[history upper label,anchor=east] (CCKSLabel)
  at (\HistoryXCCKS,0.39) {\cite{CarolanChildsKovacsDeakSchaeffer2025}};
\node[history upper label,anchor=east] (WuLabel)
  at (28,0.55) {\cite{wu2026matrixfree}};
\draw[historyUpper,history leader] (WuLabel.east)
  -- (\HistoryXWu,0.55) -- (Wu);
\node[history upper label,anchor=north east]
  at (\HistoryXThisWork,0.21) {This work};

\begingroup
\hypersetup{citecolor=black}
\node[history lower label,anchor=north west]
  at (.7,-.01)
  {\cite{BuhrmanDeWolf1999,FarhiGoldstoneGutmannSipser1998}};
\node[history lower label,anchor=south]
  at (1.35,0.09) {\cite{Ambainis1999}};
\node[history lower label,anchor=south]
  at (\HistoryXHNS,0.225) {\cite{HoyerNeerbekShi2002}};
\endgroup

\draw[history upper] (17.0,1.045) -- (18.6,1.045);
\node[anchor=west] at (18.75,1.045) {Upper bounds: exact};
\draw[history expected] (17.0,0.983) -- (18.6,0.983);
\node[anchor=west] at (18.75,0.983) {Upper bound: expected zero-error};
\draw[history lower] (17.0,0.921) -- (18.6,0.921);
\node[anchor=west] at (18.75,0.921) {Lower bounds};
\end{tikzpicture}
    \caption{History of quantum algorithms and lower bounds for ordered search.}
    \label{fig:history}
\end{figure}

Notably, most known quantum algorithms for ordered search are obtained by numerically solving a fixed
finite instance and recursing~\cite{FarhiGoldstoneGutmannSipser1999,BrookesJacokesLandahl2004,ChildsLandahlParrilo2007,CarolanChildsKovacsDeakSchaeffer2025,wu2026matrixfree}, which yields query complexity
$k\log_mn=\frac{k}{\ln m}\ln n$ from a finite solution to the size-$m$ problem using $k$ queries. 
Approaching the optimal constant by this route would likely require optimally solving a sequence of finite instances of unbounded size, while
numerical methods already strain at $k=6$~\cite{wu2026matrixfree}.

The zero-error algorithm of Ben-Or and Hassidim~\cite{BenOrHassidim2007} is also based on numerically understanding a finite instance, in this case explicitly calculating the information output of the greedy algorithm of Farhi et al.~\cite{FarhiGoldstoneGutmannSipser1999}. The greedy algorithm is then recursively used inside an error-tolerant classical ordered search algorithm. This strategy again runs into the aforementioned barrier for achieving the optimal constant.

The one exception in the literature is the approach of Hoyer, Neerbek, and Shi~\cite{HoyerNeerbekShi2002}, based on a combinatorial pebbling of a binary tree satisfying certain properties. These properties allow a quantum algorithm to  
solve ordered search by making recursive progress controlled by the number of pebbles used. Given a suitable colored pebbling of an $n$-leaf binary tree using $n'$ pebbles per color, the algorithm succeeds with a query count satisfying the recurrence $T(n) \leq T(n'+1)+1$. However, one can show that at least $n' \geq n/3$ pebbles are necessary to satisfy the constraints, which presents an obstacle to achieving optimality 
using this approach as well.

\paragraph{Our results.}
We give two new quantum algorithms for ordered search, each using
$\frac1\pi\ln n+o(\log n)$ queries.
The first (\Cref{sec:zero-error}) is a zero-error algorithm with
expected query complexity $\frac1\pi\ln n+O\bigl((\ln n)^{3/4}\bigr)$
(\Cref{cor:continuum-query-complexity}).  The algorithm is simple to
state: prepare a fixed, broad wave packet on the real line, alternately apply
the phase oracle in the position basis and a fixed reflection in the
Fourier basis for $\approx\frac1\pi\ln n$ rounds, and measure the position.
The second (\Cref{sec:green-mixture-ordered-search}) is an exact
algorithm using $\frac1\pi\ln n+o(\log n)$ queries in the worst case
(\Cref{thm:ordered-search}).  It is obtained by writing down an
explicit analytic solution path for the polynomial program
of~\cite{FarhiGoldstoneGutmannSipser1999}.

Combined with the lower bound of~\cite{HoyerNeerbekShi2002}, the second result determines
the exact quantum query complexity of ordered search up to lower-order
terms:
\[
    \Q_E(\textsc{OrderedSearch}_n)
      =\frac{\ln n}{\pi}+o(\log n)
      \approx0.221\log_2n.
\]
The same statement holds for the expected query complexity of
a zero-error algorithm \cite{BelovsYolcu2023}.  In particular, the adversary bound is tight for
this problem, confirming the prediction of~\cite{ChildsLee2008}, and
the optimal quantum speedup for ordered search is by the factor
$\pi\log_2e\approx4.53$.

\subsection*{Technical Overview.}

Both algorithms reduce to a one-dimensional family of states, where a certain ``logarithmic coordinate'' parameterizing this family translates by approximately $\pi$ on each query. The unsolved initial state and solved final state are both members of this family and differ by $\ln n + o(\log n)$ in the logarithmic coordinate, so $\approx \frac{1}{\pi} \ln n$ queries suffice.

\paragraph{Claude's algorithm.} Claude's zero-error algorithm arises from a continuum relaxation, taking the state space of the algorithm to be the real line. A continuum oracle can be implemented with one discrete query: coherently compute the integer floor of the continuous position, query the corresponding input bit, and uncompute the index, extending the range of the input by $0$ to the left and $1$ to the right.
In this picture, the oracle becomes a phase flip below an unknown transition point~$s$, which we denote $O_s$. For instance, $O_0$ applies a sign function, negating all components of the state that are below $\ket{0}$. 

The algorithm first prepares a fixed initial state, a wide wave packet on the scale of $n^{1+o(1)}$. This state is approximately evenly spread over transition points $s \in [n]$, reflecting uncertainty over the marked location.\footnote{Formally, we take a packet whose log-position profile is a Gaussian centered at $\ln n + o(\log n)$ and width $\omega(1)$ but $o(\log n)$, though the algorithm would likely work with other choices.} After preparing this initial state, the algorithm repeatedly applies a \emph{driving} unitary $U_s$, which can be implemented in a single query. Specifically, $U_s$ is the composition of the oracle query $O_s$, followed by the sign function $O_0$ applied in the Fourier basis.

We can understand the action of $U_s$ using two symmetries: \begin{enumerate}
    \item Translation invariance. The algorithm's initial state is approximately invariant under translations due to its spread, and the non-oracular component of $U_s$ is exactly invariant, as it is diagonal in the Fourier basis. The oracle $O_s$ is equivariant under translations, in the sense that translating the input by $t$ is equivalent to translating the marked position by $-t$. One can use these facts to show that the success of the algorithm on $s=0$ implies the success of the algorithm on any other marked element.
    \item Scale invariance. The driving operator $U_0$ commutes with dilations of the real number line. The Mellin transform, which is the Fourier transform of log-position, diagonalizes parity-preserving scale-invariant operators. Moreover, due to our choice of a broad initial state, the state in the Mellin basis is highly concentrated near the origin. This concentration is maintained by the algorithm, so we can understand the dispersion relation by a simple first-order Taylor expansion.
\end{enumerate}

Carrying out this calculation, we find that that on the even subspace, the action of $U_0$ near Mellin frequency zero is approximately a phase shear: the mode $\omega$ obtains a phase $\pi \omega$ (modulo an unimportant global phase). As this basis is the Fourier transform of log-position, in log-position this corresponds to a negative translation by $\pi$. This provides a simple characterization of the dynamics: the initial wave packet is concentrated near $\ln n$ in log-position space, and is negatively shifted by $\pi$ on each query. When the packet reaches negative log-position coordinates, this corresponds to a position packet sharply concentrated near zero,\footnote{This is essentially the fact that $\lim_{x \rightarrow -\infty} e^x = 0$.} which represents success of the algorithm.

To summarize, the state of the algorithm is a wave packet in log-position space. This packet begins centered near $\ln n$, representing ignorance of the marked element. As the algorithm makes progress, the center moves to the left, eventually reaching negative values representing strong knowledge of the marked element. Each driving operation moves the packet a distance $\pi$, so $\frac{1}{\pi} \ln n + o(\log n)$ queries suffice to solve ordered search. An example of two translated wave packets in log-position space is depicted in \Cref{fig:log-packets}.

\paragraph{GPT's algorithm.} GPT's exact algorithm directly solves a polynomial program introduced by Farhi et al.~\cite{FarhiGoldstoneGutmannSipser1999}, which they showed implies quantum algorithms for ordered search. The program is a sequence of Laurent polynomials that begins with the so-called Fej\'er kernel, ends with the constant polynomial, and enforces two constraints at each step: \begin{enumerate}
    \item Each Laurent polynomial in the sequence is nonnegative on the unit circle.
    \item A constraint that depends on the parity of the index of the step, with odd-to-even steps fixing the reversal-symmetric (i.e., symmetric under reversing the non-constant coefficients) component and even-to-odd steps fixing the reversal-antisymmetric component.
\end{enumerate}

The solution is built from a one-parameter family of exactly nonnegative Laurent
polynomials, the \emph{solution family} $G_{\lambda,n}$, with hyperbolic
coefficient profile $\sinh(\lambda(1-j/n))/\sinh\lambda$ on the degree-$j$ term. We refer to $\lambda$ as the \emph{rate}.
The solution family includes the Fej\'er and constant polynomials as endpoints, with $\lambda=0$ the Fej\'er kernel and $\lambda=\infty$ the constant polynomial. The sequence corresponding to the algorithm follows a
greedy rule: beginning from the Fej\'er kernel at $\lambda = 0$, split the family into the reversal-symmetric and reversal-antisymmetric component and at each step push the free component as far up in $\lambda$ as nonnegativity allows. Note that the reversal-symmetric and reversal-antisymmetric components come from different polynomials in the family, in particular having different values of $\lambda$: if (suppressing $n$) we have $1 + S_\lambda + A_\lambda = G_\lambda$ for reversal-symmetric $S_\lambda$ and reversal-antisymmetric $A_\lambda$, then the next polynomial in the greedy recurrence will be of the form\footnote{Technically, the polynomial will actually be a convex combination of polynomials of this form.} $1+S_\beta+A_\alpha$. Nonetheless, both $\alpha$ and $\beta$ will approach $\infty$ as the algorithm progresses.

Analyzing this strategy requires showing that the $\alpha, \beta$ variables increase fast enough.  
Taking a suitable mixture of components for the packet,
for each fixed $d<\pi$ sufficiently close to $\pi$,
we show that successive packet centers increase
by at least a factor $e^d$ above a parameter-dependent threshold.
Choosing the mixture parameters to vary sufficiently slowly with $n$
allows $d\to\pi$ while keeping the initialization cost $o(\log n)$.
Once every rate in both component mixtures is at least $2n$, two
final queries suffice to exactly reach the polynomial $1$, making
the algorithm exact.
These results imply that the solution family uses $\frac{1}{\pi} \ln (n) + o(\log n)$ steps.

Showing this mathematically involves reasoning about the nonnegativity of polynomials of the form $1 + S_\beta + A_\alpha$. We first convert nonnegativity into an algebraically cleaner determinant condition on the unit circle. We then use periodization to rewrite the condition in terms of an even simpler convolution. Using this simplified form, we show that a compactly supported mixture with quarticly decaying edges achieves any factor $e^d$ with $d<\pi$
(\Cref{prop:compact-mixtures}).  Endpoint arguments at rate $0$ and at
rates beyond $n$ complete an exact path of length
$\frac1\pi\ln n+o(\log n)$.

The two constructions are informally similar: the
zero-error algorithm moves a wave packet at speed $\pi$ in
log-\emph{position}, while the exact algorithm moves a certificate
packet at speed $\pi$ in log-\emph{rate}. While the exact algorithm is superior in an algorithmic sense, the zero-error algorithm is very simple, and the driving unitary follows naturally from asserting both scale and translation invariance after passing to the continuum picture.

\paragraph{Open problems.}
Our results leave the lower-order behavior open.  The exact complexity
is $\frac1\pi\ln n+o(\log n)$; is it $\frac1\pi\ln n+O(1)$, as the
lower bound $\frac1\pi(\ln n-1)$ permits?  Relatedly, our exact
algorithm is analytic but not explicitly optimized; it would be
interesting to extract algorithms for finite $n$ and compare against prior works
If we allow error at most $\epsilon>0$, the best $\epsilon$-dependent constant in the complexity of ordered search remains open. In particular, the adversary lower bound has a fixed dependence on $\epsilon$ that may not be achievable.
Finally, our algorithms are query-efficient;
we did not attempt to optimize time complexity, and in particular the
exact algorithm relies on Fej\'er--Riesz factorization of the
certificate polynomials to produce unitaries.

\paragraph{Organization.}
\Cref{sec:background} recalls the ordered search problem and the
polynomial formulation.  \Cref{sec:zero-error} presents the
zero-error continuum algorithm.
\Cref{sec:green-mixture-ordered-search} presents the exact algorithm.

\bigskip

\noindent\fbox{%
  \begin{minipage}{\dimexpr\linewidth-2\fboxsep-2\fboxrule}
    \paragraph{AI Disclosure.}
    The core technical results in this paper were discovered autonomously by AI systems. All proofs have been verified and refined by the authors, who take sole responsibility for any mistakes. The manuscript was prepared with assistance from GPT-5.6-Sol, GPT-6-Astra, and Claude Fable~5, and edited by the authors.
  \end{minipage}%
}


\section{Background}
\label{sec:background}

\subsection{Ordered search}
\label{sec:ordered-search}

In the ordered search problem we are given a sorted list
$a_1\le\cdots\le a_n$ together with a target value $t$ that is promised to occur in the list. We are tasked with finding the first occurrence of $t$, using comparisons of the form ``$a_i\ge t$?''. Equivalently, we can take the input as $x=0^j1^{n-j}\in\bit^n$ for an unknown $j\in\{0,\ldots,n-1\}$, a query
at index $i$ returns $x_i=\mathbf 1[i\ge j]$, and the goal is to
output~$j$.  We refer to $j$ as the \emph{transition point} of the
input.  In \Cref{sec:zero-error} it will be convenient to center the
instance, relabeling indices so that the transition point $s$
satisfies $s\in[-n/2,n/2]$ and the query at position $i$ returns
$\mathbf 1[i\ge s]$; the two conventions differ only notationally.

\begin{problem}[Ordered search]
\label{prob:ordered-search}
For $n\ge1$, let
\[
    \mathsf{OSP}_n
      \coloneqq\set{0^j1^{n-j}:0\le j\le n-1}\subseteq\bit^n.
\]
Given query access to an input $x=0^j1^{n-j}\in\mathsf{OSP}_n$, where a
query at index $i\in\{0,\ldots,n-1\}$ returns
$x_i=\mathbf 1[i\ge j]$, output the transition point~$j$.
\end{problem}

In the quantum query model, access to $x$ is provided by the phase
oracle 
\[
    O_x\ket i=(-1)^{x_i}\ket i,
\]
and a $k$-query algorithm is a sequence of input-independent unitaries
interleaved with $k$ applications of $O_x$, followed by a measurement. There is no loss of generality in not providing controlled access, as $x_{n-1}=1$ for all inputs $x \in \mathsf{OSP}_n$.
An algorithm is \emph{exact} if it outputs $j$ with certainty on every
input, and \emph{zero-error} if it never outputs an incorrect value but has randomized runtime;
we then count the expected number of queries.  We write
$\Q_E(\textsc{OSP}_n)$ for the minimum number of queries of an
exact algorithm.  A candidate transition point can be verified with two
classical comparison queries, a fact we use to convert high-probability
algorithms into zero-error ones.

Following~\cite{FarhiGoldstoneGutmannSipser1999}, it is sometimes useful to pass
to a symmetrized variant.  Replace the input $x$ by the $2n$-bit string
$y=x\bar x$, where $\bar x$ is the bitwise complement of $x$.  The
resulting inputs, together with their bitwise complements, are exactly the $2n$ cyclic shifts of
$1^n0^n$; the problem asks for the shift modulo $n$; and the input
family is equivariant under the cyclic translation
$T\ket i=\ket{i+1\bmod 2n}$.  The symmetrized and original problems are
equivalent as query problems, with no overhead in either
direction~\cite{FarhiGoldstoneGutmannSipser1999,CarolanChildsKovacsDeakSchaeffer2025}.  A quantum algorithm for the
symmetrized problem is \emph{translation invariant} if it acts on the
single $2n$-dimensional index register with no workspace, its initial
state is the uniform superposition, and each of its non-query unitaries
commutes with $T$, i.e., is diagonal in the Fourier basis of
$\mathbb Z/2n$~\cite{FarhiGoldstoneGutmannSipser1999}.

Translation-invariant algorithms were introduced because they are
analytically tractable: exact algorithms in this class are
characterized by a feasibility program over nonnegative Laurent
polynomials, stated as \Cref{prop:ordered-search} below.  A priori this
characterization only yields upper bounds.  However, Carolan, Childs,
Kovacs-Deak, and Schaeffer recently proved that the restriction is
without loss of generality: any $k$-query quantum algorithm for
ordered search can be converted into a translation-invariant,
$k$-query algorithm  with the same success probability~\cite[Theorem~1]{CarolanChildsKovacsDeakSchaeffer2025}.  Consequently the polynomial
program computes $\Q_E$ exactly, and \Cref{prop:ordered-search} below
is stated in this stronger, unconditional form: the forward direction
(from a polynomial path to an algorithm) is due
to~\cite{FarhiGoldstoneGutmannSipser1999} via Fej\'er--Riesz
factorization~\cite{GrenanderSzego1958,ChildsLandahlParrilo2007}, and the converse is
\cite[Theorem~1]{CarolanChildsKovacsDeakSchaeffer2025}.  The same
equivalence implies that feasibility of the polynomial program is monotone in $n$: a $k$-query
exact algorithm for lists of length $n$ yields one for every shorter
list~\cite[Corollary~24]{CarolanChildsKovacsDeakSchaeffer2025}.

\subsection{The polynomial formulation}
\label{subsec:poly-form}

Write a normalized symmetric Laurent polynomial $q$ of degree at most $n-1$ as $q(z)=1+\sum_{j=1}^{n-1}a_j(z^j+z^{-j})$ where we take $|z|=1$. Let $J\colon\mathbb R^{n-1}\to\mathbb R^{n-1}$ reverse coordinates (i.e., $(Ja)_j=a_{n-j}$),
and let
\[
    P_+=\frac{I+J}{2},
    \qquad
    P_-=\frac{I-J}{2}.
\]
The initial polynomial is the normalized Fej\'er kernel $F_n(z)\coloneqq\frac1n\left|1+z+\cdots+z^{n-1}\right|^2$,
whose coefficient vector is $a^{(0)}_j=1-\frac{|j|}{n}$. We use the standard translation-invariant characterization: 

\begin{proposition}[\cite{FarhiGoldstoneGutmannSipser1999,CarolanChildsKovacsDeakSchaeffer2025}]
    \label{prop:ordered-search}
    There exists an exact $k$-query quantum algorithm for $n$-element ordered search if and only if there exist nonnegative symmetric Laurent polynomials $q_0,\ldots,q_k$ of degree less than $n$ satisfying
    \begin{equation}
        q_0=F_n,
        \qquad
        q_k=1, \qquad \int_{0}^{2\pi} q_t(e^{i\theta})d\theta = 2\pi \qquad (\forall t \in [k]),
        \label{eq:path-endpoints}
    \end{equation}
    and
    \begin{equation}
        J\bigl(a^{(t)}-a^{(t-1)}\bigr)
          =(-1)^{t+1}\bigl(a^{(t)}-a^{(t-1)}\bigr)
        \qquad (1\le t\le k).
        \label{eq:alternating-parity}
    \end{equation}
\end{proposition}

Equivalently, successive increments alternate between the two eigenspaces of
$J$.  Condition~\eqref{eq:alternating-parity} is also equivalent to agreement
of $q_t$ and $q_{t-1}$ on the roots satisfying $z^n=(-1)^t$.

\subsection{Transforms}

Claude's algorithm utilizes a number of transforms on $\mathcal H = L^2(\mathbb R)$, the Hilbert space of square-integrable complex functions on the real line.

\begin{definition}[Fourier transform]
\label{def:fourier-trans}
    The unitary Fourier transform of $\psi \in \mathcal H$ is denoted
\[
    \mathcal F\psi(\omega) = \hat \psi(\omega) \coloneqq \frac{1}{\sqrt{2\pi}}
        \int_{\mathbb R}e^{-i\omega x}\psi(x)\,dx.
\]
\end{definition} 
Translations are denoted by $(T_a\psi)(x)=\psi(x-a)$. A linear operator on $\mathcal H$ that commutes with $T_a$ for all $a$ is called \emph{translation invariant}, and is diagonalized by the Fourier transform.

Let $\mathcal H_+\subseteq \mathcal H$ be the subspace of even functions and $\mathcal H_- \subseteq \mathcal H$ the subspace of odd functions, such that $\mathcal H = \mathcal H_+ \oplus \mathcal H_-$ (an internal direct sum). The log-position transform is most easily defined by its inverse transform $\mathcal L^\dagger \colon \mathcal H \oplus \mathcal H \rightarrow \mathcal H$ (where the domain of $\mathcal L^\dagger$ is the external direct sum $\mathcal H \oplus \mathcal H \cong \mathcal H \otimes \C^2$).
This operator splits into $\mathcal L_+^\dagger \colon \mathcal H \rightarrow \mathcal H_+$ and $\mathcal L_-^\dagger \colon \mathcal H \rightarrow \mathcal H_-$, i.e., we can write $\mathcal L^\dagger(h_+, h_-) = \mathcal L_+^\dagger(h_+) + \mathcal L_-^\dagger(h_-)$ according to the following.

\begin{definition}[Inverse log-position transform]
    The inverse log-position transform $\mathcal L^\dagger$ acting on $h_+, h_- \in \mathcal H$ is defined in terms of $\mathcal L_+^\dagger$ and $\mathcal L_-^\dagger$ as above, where \[
    \mathcal L_+^\dagger h(x)
      \coloneqq\frac{1}{\sqrt 2}|x|^{-1/2}h(\ln|x|), \quad \mathcal L_-^\dagger h(x) \coloneqq\frac{\operatorname{sign}(x)}{\sqrt 2}|x|^{-1/2}h(\ln|x|),\]
     for all $x \neq 0$.
     \label{def:log-trans}
\end{definition} One can verify that $\mathcal L_+^\dagger$ and $\mathcal L_-^\dagger$ are unitary, which implies that $\mathcal L^\dagger$ is also unitary. We can now define the Mellin transform as the composition of the log-position transform and Fourier transform.

\begin{definition}[Mellin transform]
    The Mellin transform $\mathcal M \colon \mathcal H \rightarrow \mathcal H \oplus \mathcal H$ is the composition of the log-position transform and the Fourier transform, $\mathcal M = (\mathcal F \oplus \mathcal F) \cdot \mathcal L$.
    \label{def:mellin-trans}
\end{definition}
We denote dilations by $(D_a \psi)(x) = \sqrt{a} \psi(a\cdot x)$. A bounded linear operator on $\mathcal H$ that commutes with $D_a$ for all $a>0$ is called \emph{dilation invariant}. Under the Mellin transform, such an operator acts by multiplication by a $2 \times 2$ matrix-valued function. If it also preserves the even and odd subspaces, this matrix is diagonal.

A key class of states will be the set of even states (i.e., $\psi$ with $\mathcal L_- \psi =0$) that are Gaussian in log-position (i.e., $\mathcal L_+ \psi$ is a Gaussian). These are defined as follows.

\begin{definition}
\label{def:log-gaussian}
    For a center $u\in\mathbb R$ and width $w>0$, define the $L^2$-normalized Gaussian
\begin{equation}
    h_{u,w}(y)
      \coloneqq\frac{1}{(2\pi w^2)^{1/4}}
        \exp\left(-\frac{(y-u)^2}{4w^2}\right)
    \label{eq:log-gaussian}
\end{equation}
and the corresponding log-Gaussian even packet
\begin{equation}
    \psi_{u,w}\coloneqq\mathcal L_+^\dagger h_{u,w}.
    \label{eq:physical-log-gaussian}
\end{equation}
\end{definition}

\begin{figure}
    \centering
\includegraphics[width=0.8\linewidth]{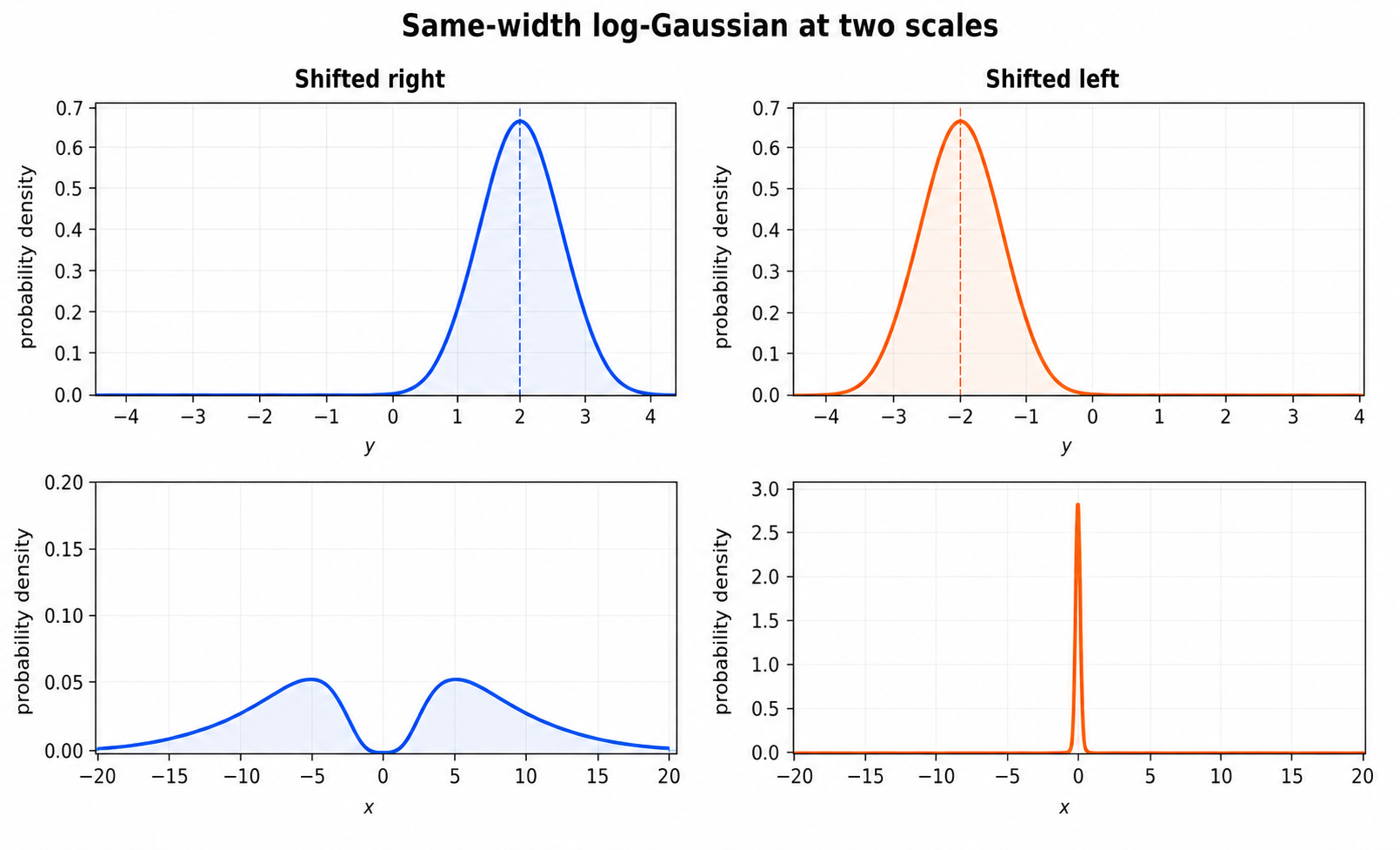}
    \caption{Log-position Gaussian packets, where $x$ is position and $y$ is log-position. Note that shifting the packet left in log-position corresponds to concentrating the packet in position.}
    \label{fig:log-packets}
\end{figure}

As depicted in \Cref{fig:log-packets}, translations in log-position space correspond to dilations in position space. As the Gaussian wave packet translates from right to left in log-position space, the position-space wave function contracts sharply around the origin. Taking the position-space origin as the oracle transition point, Claude's algorithm will begin with a log-Gaussian centered near $\ln n$ in log-position space, i.e., a wave function that is widely spread in position space. The algorithm will drive this Gaussian leftwards in log-position a distance $\pi$ per query, so the packet will concentrate on the transition point after approximately $\frac{1}{\pi} \ln n$ queries.

\section{Zero-error solution}
\label{sec:zero-error}
In this section, we describe and analyze a zero-error quantum algorithm for ordered search that uses $\frac{1}{\pi} \ln n + o(\log n)$ queries in expectation.
The algorithm maps the discrete ordered search problem described in~\Cref{prob:ordered-search} to a continuum variant. Take the algorithm's workspace to be the Hilbert space $\mathcal H=L^2(\mathbb R)$. If we center the ordered-search instance so that the unknown transition point satisfies $s\in[-n/2,n/2]$, then the phase oracle $(O_s\psi)(x)=(-1)^{\mathbf 1[x\ge s]}\psi(x)$ can be implemented straightforwardly with one query to the discrete oracle. The algorithm is:
\begin{enumerate}
    \item Prepare an initial state that is Gaussian in log-position, $\psi_{\mathrm{in}}=\psi_{u,w}$, with center $u=\ln n + \ln^{3/4} n$ and width $w = \ln^{1/2} n$.
    \item Alternately apply the oracle $O_s$ in the standard basis followed by the non-oracular operation $O_0$ in the Fourier basis, for $Q=
      \left\lceil
        \frac{\ln n+2\ln^{3/4} n}{\pi}
      \right\rceil$ iterations.
\end{enumerate}

The specific choice of input state is not particularly important, and there are likely other states that would work. The key feature of $\psi_{\mathrm{in}}$ is that it is sharply concentrated around the origin in the Mellin basis, which follows from a large spread in log-position space. This is consequence of the fact that the Fourier transform of a Gaussian with width $w$ is a Gaussian with width $O(1/w)$. For an appropriately chosen $\psi_{\mathrm{in}}$, measuring the final state in the position basis and rounding to the nearest integer will give $s$ with $1-o(1)$ probability. To obtain a zero-error algorithm, one can simply repeat the procedure until $s$ is identified.

Note that the operator in step two can be applied by a single query, as $O_0$ is a fixed oracle-independent operator. We sometimes write the operator in step two as $U_s=\mathcal F^\dagger O_0\mathcal F O_s$, and refer to it as the \emph{driving} operator.

\begin{remark}
    \label{rem:trans-scale-invar}
    The choice of $\mathcal F^\dagger O_0 \mathcal F$ in the driving operator is both translation- and scale-invariant. The only other such operators are the identity, and linear combinations of the identity and $\mathcal F^\dagger O_0 \mathcal F$.
\end{remark}

In the continuum picture instances are translates of each one another, where instances of size $n$ are within translations by $O(n)$. Therefore, translation equivariance of the driving operation and approximate translation invariance of the initial state mean that an algorithm that works for one marked element will work for all marked elements. Scale invariance, on the other hand, is the phenomenon leading to exponential growth in the progress of the algorithm, and therefore a logarithmic number of queries. The uncertainty of the algorithm decreases by a constant factor approaching $e^{-\pi}$ with each query, which leads to scale-invariant behavior. These facts make the choice of driving operation extremely natural.

\subsection{Initial state}
\label{subsec:params-choice}
To construct the initial state, we write
\begin{equation}
    L\coloneqq\ln n,
    \qquad
    w\coloneqq L^{1/2},
    \qquad
    r\coloneqq L^{3/4},
    \qquad
    u_{\mathrm{in}}\coloneqq L+r,
    \label{eq:continuum-parameters}
\end{equation}
where $u_{\mathrm{in}}$ parametrizes the initial state as $\psi_{\mathrm{in}}\coloneqq\psi_{u_{\mathrm{in}},w}$.
Thus the initial packet has log-position width $o(\ln n)$, is centered at
$\ln n+o(\ln n)$, and has Mellin-frequency width $O((\ln n)^{-1/2})$. We will show that this packet is nearly invariant under any translation of magnitude less than $n/2$. Combined with the translation equivariance of the algorithm, $U_s = T_s U_0 T_s^\dagger$, this implies that the success probability with transition point at $0$ is within $o(1)$ of the success probability for any other transition point.

\begin{lemma}[Uniform approximate translation invariance of the initial packet]
\label{lem:initial-packet-translation-invariance}
Uniformly for $|s|\le n/2$,
\[
    \left\|
        T_s\psi_{\mathrm{in}}-\psi_{\mathrm{in}}
    \right\|_2
    =o(1).
\]
\end{lemma}

\begin{proof}
Recall that $\psi_{\mathrm{in}}=\psi_{u_{\mathrm{in}},w}$ 
has log-position mean $u_{\mathrm{in}}$ and standard deviation $w$.
Since $u_{\mathrm{in}}=\ln n+r$,
this scale is $e^{u_{\mathrm{in}}}=ne^r = \omega(n)$.
Thus a translation by $|s|\le n/2$ is negligible compared with the scale
on which the packet varies. Explicitly, if $X$ denotes position measured in
$\psi_{\mathrm{in}}$, then $Y\coloneqq\ln|X|$ is Gaussian with mean $u_{\mathrm{in}}$ and variance $w^2$. Hence, taking $R\coloneqq e^{u_{\mathrm{in}}-r/2}=ne^{r/2}$, we have
\[
    \Pr[|X|\le R]
    =
    \Pr\!\left[Y-u_{\mathrm{in}}\le-\frac r2\right]
    =o(1).
\]
Thus, up to $o(1)$ in norm, $\psi_{\mathrm{in}}$ is supported on
$|x|\ge R$.

On this region, differentiating the Gaussian packet shows that
\[
    |\psi_{\mathrm{in}}'(x)|
    =
    O\!\left(\frac{1}{|x|}\right)
    \left(1+\frac{|\ln|x|-u_{\mathrm{in}}|}{w^2}\right)
    |\psi_{\mathrm{in}}(x)|,
\]
and therefore
\[
    \|\psi_{\mathrm{in}}'\|_{L^2(|x|\ge R)}
    =O(R^{-1}).
\]
Consequently, for $|s|\le n/2$,
\[
    \|T_s\psi_{\mathrm{in}}-\psi_{\mathrm{in}}\|_2
    \le o(1)+O\!\left(\frac{|s|}{R}\right)
    \le o(1)+O(e^{-r/2})
    =o(1).
\]
This is uniform in $|s|\le n/2$.
\end{proof}

The total number of queries is $Q\coloneqq\left\lceil\frac{L+2r}{\pi}\right\rceil$. In the following sections we show that a query corresponds to a negative translation by $\pi$, motivating the notation $u_{\mathrm{out}}\coloneqq u_{\mathrm{in}}-\pi Q$.
$Q$ is chosen such that $-r-\pi<u_{\mathrm{out}}\le-r$.

\subsection{Algorithm analysis}

In this section, we establish that the query operator $U_0$ is close to a negative translation by $\pi$ in log-position space on any even state concentrated near Mellin frequency $\omega=0$. This suffices to understand the behavior for any $U_s$ due to  translation equivariance, as described in \Cref{subsec:params-choice}. The approximation is such that $t$ applications of $U_0$ are well approximated by $t$ such dilations, for any $t = o(w^3)$. The global phase accumulation $e^{-i\pi t/2}$ is irrelevant.

\begin{lemma}[Translation of an even Mellin packet]
\label{lem:mellin-packet-translation}
For the log-Gaussian packet $\psi_{u, w}$,
\begin{equation}
    \left\|
        U_0^t\psi_{u,w}
        -e^{-i\pi t/2}\psi_{u-\pi t,w}
    \right\|_2
    =O\left(\frac{t}{w^3}\right),
    \label{eq:packet-propagation}
\end{equation}
uniformly in the center $u$. Thus, whenever $t=o(w^3)$, the packet is
translated to the left by $\pi t$ in log-position up to vanishing error.
\end{lemma}

\begin{proof}

We actually establish the following more general fact: there is a universal constant $C>0$ such that, for every integer $t\ge 0$
and every $h\in L^2(\mathbb R)$ satisfying
$|\omega|^3\widehat h(\omega)\in L^2(\mathbb R)$,
\begin{equation}
    \left\|
        U_0^t\mathcal L_+^\dagger h
        -e^{-i\pi t/2}\mathcal L_+^\dagger T_{-\pi t}h
    \right\|_2
    \le
    Ct\left\||\omega|^3\widehat h(\omega)\right\|_2.
    \label{eq:general-packet-propagation}
\end{equation}

On the even subspace, the Mellin transform diagonalizes $U_0$. We will show that the eigenvalue $\lambda(\omega)=e^{-i\theta(\omega)}$ of mode $\omega$ is
\[
    \lambda(\omega) = e^{-i\theta(\omega)} = \tanh(\pi\omega)-i\operatorname{sech}(\pi\omega),
    \qquad
    \theta(\omega)\coloneqq 2\arctan(e^{-\pi\omega}).
\]
Indeed, for the generalized even and odd Mellin modes
\[
    \phi_\omega^+(x)
      \coloneqq\frac{1}{\sqrt2}|x|^{-1/2+i\omega},
    \qquad
    \phi_\omega^-(x)
      \coloneqq\frac{\operatorname{sgn}(x)}{\sqrt2}
        |x|^{-1/2+i\omega},
\]
the identities
\begin{align*}
    \int_0^\infty x^{z-1}\cos(\xi x)\,dx
      &=\Gamma(z)\cos\left(\frac{\pi z}{2}\right)\xi^{-z},\\
    \int_0^\infty x^{z-1}\sin(\xi x)\,dx
      &=\Gamma(z)\sin\left(\frac{\pi z}{2}\right)\xi^{-z}
\end{align*}
give, upon setting $z=\frac12+i\omega$,
\begin{align*}
    \mathcal F\phi_\omega^+
      &=a_+(\omega)\phi_{-\omega}^+,
    &
    a_+(\omega)
      &=\sqrt{\frac{2}{\pi}}\,
        \Gamma\left(\frac12+i\omega\right)
        \cos\left(\frac{\pi}{4}+\frac{i\pi\omega}{2}\right),\\
    \mathcal F\phi_\omega^-
      &=a_-(\omega)\phi_{-\omega}^-,
    &
    a_-(\omega)
      &=-i\sqrt{\frac{2}{\pi}}\,
        \Gamma\left(\frac12+i\omega\right)
        \sin\left(\frac{\pi}{4}+\frac{i\pi\omega}{2}\right).
\end{align*}
Letting $S$ denote the sign multiplier, since $U_0=\mathcal F^\dagger S\mathcal F S$ and
$S\phi_\omega^\pm=\phi_\omega^\mp$, it follows that
\[
    U_0\phi_\omega^+
      =\frac{a_-(\omega)}{a_+(\omega)}\phi_\omega^+
      =\left(
          \tanh(\pi\omega)-i\operatorname{sech}(\pi\omega)
        \right)\phi_\omega^+.
\]
Here the Gamma factors cancel, and we used $-i\tan\left(\frac{\pi}{4}+iz\right)=\tanh(2z)-i\operatorname{sech}(2z)$.

Near $\omega=0$, we have $\theta(\omega)=\frac{\pi}{2}-\pi\omega+\frac{\pi^3}{6}\omega^3+O(\omega^5)$.
On the other hand, $\widehat{T_{-\pi t}h}(\omega)=e^{i\pi t\omega}\widehat h(\omega)$, so $e^{-i\pi t/2}T_{-\pi t}$ has Mellin multiplier $e^{-it(\pi/2-\pi\omega)}$.
Put $R(\omega) \coloneqq\theta(\omega)-\left(\frac{\pi}{2}-\pi\omega\right).$
The Taylor expansion gives $R(\omega)=O(\omega^3)$ near the origin. Since
$0<\theta(\omega)<\pi$, after increasing the constant this becomes the
global bound $|R(\omega)|\le C|\omega|^3$.
Using $|e^{i\alpha}-e^{i\beta}|\le|\alpha-\beta|$ and Plancherel's theorem,
we obtain
\begin{align*}
    \left\|
        U_0^t\mathcal L_+^\dagger h-e^{-i\pi t/2}\mathcal L_+^\dagger T_{-\pi t}h
    \right\|_2
    &\le t\|R(\omega)\widehat h(\omega)\|_2\\
    &\le Ct\||\omega|^3\widehat h(\omega)\|_2.
\end{align*}

Finally, $|\widehat h_{u,w}(\omega)|^2$ is a centered Gaussian density of
variance $1/(4w^2)$, and hence
\[
    \int_{\mathbb R}\omega^6
      |\widehat h_{u,w}(\omega)|^2\,d\omega
      =\frac{15}{64w^6}.
\]
Therefore
\[
    \||\omega|^3\widehat h_{u,w}\|_2
      =\frac{\sqrt{15}}{8w^3}.
\]
Since $T_{-\pi t}h_{u,w}=h_{u-\pi t,w}$, the Gaussian specialization
follows.
\end{proof}

Using \Cref{lem:initial-packet-translation-invariance} and \Cref{lem:mellin-packet-translation}, we can now establish correctness of the algorithm.

\begin{lemma}[Uniform localization after propagation]
\label{lem:uniform-final-localization}
Let $X_s$ be the outcome obtained by measuring
$U_s^Q\psi_{\mathrm{in}}$ in the position basis. Uniformly for $|s|\le n/2$,
\begin{equation}
    \Pr\left[|X_s-s|<\frac12\right]
    =
    1-O\left((\ln n)^{-1/2}\right).
    \label{eq:uniform-recovery-probability}
\end{equation}
Thus rounding $X_s$ to the nearest possible transition point recovers $s$
with probability $1-o(1)$.
\end{lemma}

\begin{proof}
Observe that $U_s^Q=T_sU_0^QT_s^\dagger$ by the translation equivariance of the $U_s$. The estimate in the proof of \cref{lem:initial-packet-translation-invariance} and unitarity therefore imply
\begin{align}
    \left\|
        U_s^Q\psi_{\mathrm{in}}
        -T_sU_0^Q\psi_{\mathrm{in}}
    \right\|_2
    &=
    \left\|
        U_0^QT_s^\dagger\psi_{\mathrm{in}}
        -U_0^Q\psi_{\mathrm{in}}
    \right\|_2 \notag\\
    &=
    \left\|
        T_s^\dagger\psi_{\mathrm{in}}
        -\psi_{\mathrm{in}}
    \right\|_2 \notag\\
    &=
    O\left(e^{-\sqrt L/16}\right).
    \label{eq:replace-translated-input}
\end{align}

Applying \cref{lem:mellin-packet-translation} with $t=Q$ gives
\begin{equation}
    \left\|
        U_0^Q\psi_{\mathrm{in}}
        -e^{-i\pi Q/2}\psi_{u_{\mathrm{out}},w}
    \right\|_2
    =
    O\left(\frac{Q}{w^3}\right)
    = O(L^{-1/2})
    \label{eq:chosen-packet-propagation}
\end{equation}
since $Q=O(L)$ and $w^3=L^{3/2}$.
Combining this with \eqref{eq:replace-translated-input} yields, uniformly in
$s$,
\begin{equation}
    \left\|
        U_s^Q\psi_{\mathrm{in}}
        -
        e^{-i\pi Q/2}T_s\psi_{u_{\mathrm{out}},w}
    \right\|_2
    =
    O(L^{-1/2}).
    \label{eq:actual-to-ideal-final-state}
\end{equation}

This is the final state of the algorithm. It now remains to understand its localization. If position is measured in
$T_s\psi_{u_{\mathrm{out}},w}$, then
\[
    \ln|X-s|\sim N(u_{\mathrm{out}},w^2).
\]
Using $u_{\mathrm{out}}\le-r$,
\begin{align}
    \Pr\left[|X-s|\ge\frac12\right]
      &=
      \Pr\left[\ln|X-s|\ge-\ln2\right] \notag\\
      &\le
      \exp\left(
        -\frac{(r-\ln2)^2}{2w^2}
      \right)
      =
      e^{-\Omega(\sqrt L)}.
    \label{eq:ideal-final-tail}
\end{align}
Finally, measurement probabilities of two unit vectors differ by at most a
constant times their Euclidean distance. Equations
\eqref{eq:actual-to-ideal-final-state} and
\eqref{eq:ideal-final-tail} therefore give
\[
    \Pr\left[|X_s-s|\ge\frac12\right]
      =
      O(L^{-1/2}),
\]
which proves the lemma.
\end{proof}

\begin{corollary}[Zero-error query complexity]
\label{cor:continuum-query-complexity}
There is a zero-error ordered-search algorithm with expected quantum query complexity
\[
    \frac{\ln n}{\pi}+o(\ln n).
\]
\end{corollary}

\begin{proof}
Each application of $U_s$ uses one query, and
\begin{align*}
    Q
      &=
      \left\lceil
        \frac{\ln n+2(\ln n)^{3/4}}{\pi}
      \right\rceil
      =
      \frac{\ln n}{\pi}
      +O\left((\ln n)^{3/4}\right)
      =
      \frac{\ln n}{\pi}+o(\ln n).
\end{align*}
By \cref{lem:uniform-final-localization}, measuring and rounding produces
the correct transition point with probability
\[
    p_n=1-O\left((\ln n)^{-1/2}\right)=1-o(1).
\]
The proposed transition point can be checked exactly with a constant number of comparison queries. If verification fails, restart the procedure. The resulting
algorithm is zero-error, and its expected number of queries is
\[
    \frac{Q+O(1)}{p_n}
      =
      \frac{\ln n}{\pi}+o(\ln n). \qedhere
\]
\end{proof}

\section{Exact solution}
\label{sec:green-mixture-ordered-search}

As an alternative to the algorithm described in \Cref{sec:zero-error}, in this section we give an exact algorithm that always uses $\frac{1}{\pi} \ln n + o(\log n)$ queries. The algorithm is constructed through the polynomial program described in \Cref{subsec:poly-form}, by exhibiting a sequence $q_0, \dots, q_k$ of nonnegative symmetric Laurent polynomials satisfying \Cref{prop:ordered-search} for $k=\frac{1}{\pi} \ln n + o(\log n)$.

We construct the solution using the \emph{solution family} $G_{\lambda,n}$ of positive Laurent polynomials labeled by a nonnegative
\emph{rate} $\lambda\geq0$.  The family is chosen so that the Fej\'er
polynomial is the $\lambda=0$ member, $G_{0, n}(z)=F_n(z)$, while taking
$\lambda\rightarrow\infty$ gives the constant polynomial~$1$, $\lim_{\lambda \rightarrow \infty} G_{\lambda, n}(z) = 1$.  The rate
$\lambda$ therefore serves as a progress measure for the algorithm, with the algorithm beginning at $\lambda=0$ and terminating at $\lambda=\infty$.

The query constraints in
\eqref{eq:alternating-parity} do not permit updating a given $G_{\lambda,n}$ to some $G_{\lambda',n}$ with $\lambda'>\lambda$ in a single query, since this would change both reversal components at
once.  Instead, we split the solution family into its reversal-symmetric and
reversal-antisymmetric components and update them alternately.  At every query
we increase $\lambda$ for the permitted component as much as possible
while preserving nonnegativity.  Thus the algorithm itself has a single
greedy rule.\footnote{Note that this is different from the greedy algorithm of Farhi, Goldstone, Gutmann, and Sipser \cite{FarhiGoldstoneGutmannSipser1999}.}

The analysis supplies lower bounds on the progress of this greedy rule.  A
uniform bound shows that a factor-three increase is always possible, and an additive increase is possible in the first move. As a technical note, we take solutions to be compact convex combinations of different rates $\lambda$, which we sometimes refer to as the packet. 
The packet is sufficiently concentrated that it nonetheless makes sense to consider its dominant central rate.
For each fixed $d<\pi$ sufficiently close to $\pi$, the rate of a
suitably chosen packet increases by at least $e^d$ per step 
once the rate is large enough.
Finally, once every rate in both component mixtures is at least $2n$, two final
queries terminate at the constant polynomial exactly. Therefore, the algorithm terminates in $\frac{1}{\pi} \ln n + o(\log n)$ queries.

\subsection{An exactly positive family}

We write the solution family, which interpolates between the Fej\'er kernel and constant polynomial, as follows.

\begin{definition}[Solution family]
For $\lambda>0$ and $1\le j<n$, let
\begin{align}
    g_{\lambda,n}(j)
      &\coloneqq\frac{\sinh\bigl(\lambda(1-j/n)\bigr)}{\sinh\lambda},
      &
    G_{\lambda,n}(e^{i\theta})
      &\coloneqq1+2\sum_{j=1}^{n-1}g_{\lambda,n}(j)\cos(j\theta).
    \label{eq:green-profile}
\end{align}
We define the family at $\lambda=0,\infty$ by its limiting values.
\end{definition}

Observe that $\lim_{\lambda\rightarrow0^+}g_{\lambda,n}(j)=1-\frac jn$, so $G_{0,n}=F_n$.  On the other hand,
$g_{\lambda,n}(j) = \Theta(e^{-\lambda j/n})$ for large $\lambda$, and hence
$G_{\infty,n}=1$.

\begin{lemma}[Exact positivity of the solution family]
\label{lem:green-positive}
For every $n\ge2$ and $\lambda\in[0,\infty]$,
$G_{\lambda,n}(e^{i\theta})$ is nonnegative for every
$\theta\in\mathbb R$.
\end{lemma}

\begin{proof}
Let
\[
    r=e^{-\lambda/n},
    \qquad
    p_r(z)=\sum_{k=0}^{n-1}r^kz^k,
    \qquad |z|=1.
\]
Since $z^{-1}=\overline z$ for $|z|=1$, we have
\begin{align*}
|p_r(z)|^2
&=p_r(z)p_r(z^{-1})\\
&=\sum_{k=0}^{n-1}r^{2k}
  +\sum_{j=1}^{n-1}
    \left(\sum_{\ell=0}^{n-1-j}r^{2\ell+j}\right)
      (z^j+z^{-j})\\
&=\frac{1-r^{2n}}{1-r^2}
  +\sum_{j=1}^{n-1}
    \frac{r^j(1-r^{2(n-j)})}{1-r^2}
      (z^j+z^{-j}).
\end{align*}
Multiplying by $(1-r^2)/(1-r^{2n})$ gives
\begin{align*}
\frac{1-r^2}{1-r^{2n}}|p_r(z)|^2
&=
1+\sum_{j=1}^{n-1}
  \frac{r^j-r^{2n-j}}{1-r^{2n}}(z^j+z^{-j})\\
&=
1+\sum_{j=1}^{n-1}
  \frac{\sinh(\lambda(1-j/n))}{\sinh\lambda}
  (z^j+z^{-j})\\
&=G_{\lambda,n}(z).
\end{align*}
This is nonnegative when $\lambda>0$. The cases
$\lambda=0,\infty$ follow by continuity.
\end{proof}

We next split the coefficient vector of the solution family, excluding the constant coefficient, into its two reversal components.

\begin{definition}
Let
\[
    S_\lambda \coloneqq P_+G_\lambda,
    \qquad
    A_\lambda \coloneqq P_-G_\lambda,
\]
and write $s_{\lambda,n}$ and $a_{\lambda,n}$ for their respective
coefficient vectors.  Their entries are
\begin{align}
    s_{\lambda,n}(j)
      &=\frac{\cosh\bigl(\lambda(1/2-j/n)\bigr)}
              {2\cosh(\lambda/2)},
      &
    a_{\lambda,n}(j)
      &=\frac{\sinh\bigl(\lambda(1/2-j/n)\bigr)}
              {2\sinh(\lambda/2)}.
    \label{eq:green-parts}
\end{align}
As above, the expressions at $\lambda=0,\infty$ are interpreted by
continuity.
Note that since $P_\pm$ act only on the nonconstant terms, $G_\lambda = 1 + S_\lambda + A_\lambda$.
\label{def:sym-antisym-green}
\end{definition}

We consider convex combinations of these components that we call \emph{packets}.  Let $Y$ be an auxiliary
real random variable with compactly supported probability density $h$.  For a center
$\alpha\in[0,\infty]$, define
\begin{align}
    S_{\alpha e^Y}(j)
       &\coloneqq \mathbb E\!\left[s_{\alpha e^Y\!,n}(j)\right],
       &
    A_{\alpha e^Y}(j)
       &\coloneqq \mathbb E\!\left[a_{\alpha e^Y\!,n}(j)\right].
    \label{eq:log-mixture}
\end{align}
The variable $Y$ only specifies a convex combination; no randomness is used
by the resulting algorithm.  Since $\ln\lambda=\ln\alpha+Y$,
the density $h$ fixes the shape of a packet in log-rate space, while changing
$\alpha$ translates the packet.  Replacing $\alpha$ by $e^d\alpha$
translates it by logarithmic distance~$d$.

For clarity, we write $q_{S_{\alpha e^Y}+A_{\beta e^Y}}$
to denote the normalized Laurent polynomial whose symmetric
coefficient vector is $S_{\alpha e^Y}$ and whose antisymmetric nonconstant
coefficient vector is $A_{\beta e^Y}$.  In particular,
\[
    q_{S_{\alpha e^Y}+A_{\alpha e^Y}}
      =\mathbb E\!\left[G_{\alpha e^Y\!,n}\right]\ge0
\]
by \cref{lem:green-positive}.

\subsection{The greedy solution family}

This subsection describes the polynomial family that will constitute our solution to the polynomial program described in \Cref{prop:ordered-search}.
Fix $n$ and a density $h$ for the remainder of the construction.  For
$\alpha\in[0,\infty]$, define the two parity-dependent transition maps
\begin{align}
    \Gamma_S(\alpha)
      &=
      \sup\left\{
          \beta\in[\alpha,\infty):
          q_{S_{\beta e^Y}+A_{\alpha e^Y}}\ge0
      \right\},
      \label{eq:greedy-S-map}\\
    \Gamma_A(\alpha)
      &=
      \sup\left\{
          \beta\in[\alpha,\infty):
          q_{S_{\alpha e^Y}+A_{\beta e^Y}}\ge0
      \right\}.
      \label{eq:greedy-A-map}
\end{align}
Here and below, polynomial inequalities must hold at every point on the
unit circle.

Starting from $\alpha_0=0$, recursively define
\begin{equation}
    \alpha_t=
    \begin{cases}
       \Gamma_S(\alpha_{t-1}),&t\ \text{odd},\\
       \Gamma_A(\alpha_{t-1}),&t\ \text{even}.
    \end{cases}
    \label{eq:greedy-scale-recurrence}
\end{equation}
The associated polynomial path is
\begin{equation}
    q_0=G_{0,n}=F_n,
    \qquad
    q_t=
    \begin{cases}
       q_{S_{\alpha_t e^Y}+A_{\alpha_{t-1}e^Y}},
          &t\ \text{odd},\\
       q_{S_{\alpha_{t-1}e^Y}+A_{\alpha_t e^Y}},
          &t\ \text{even}.
    \end{cases}
    \label{eq:greedy-polynomial-path}
\end{equation}
By construction, every $q_t$ is nonnegative.  Moreover, when $t$ is odd only the reversal-symmetric component changes from $q_{t-1}$ to $q_t$, and when $t$ is even only the reversal-antisymmetric component changes.  Thus
\eqref{eq:alternating-parity} holds automatically.

The definition has a simple interpretation: at each query, move the permitted
component as far to the right in log-rate space as possible while remaining
inside the nonnegative cone.  If one of the transition maps takes the value
$\infty$, one further query sends the other component to $\infty$ and reaches
$q=1$.  We may therefore assume in the analysis below that the scales remain
finite until the endpoint criterion is met.

It remains to show that the recurrence
\eqref{eq:greedy-scale-recurrence} advances sufficiently quickly.  We first
derive an exact all-scale nonnegativity criterion and then obtain the progress
bounds needed to solve the recurrence.

\subsection{All-scale nonnegativity conditions}

We first express a polynomial with different component rates as a quadratic form. This
gives a sufficient condition for nonnegativity. We then use that condition in two ways: to
prove a factor-three increase at every positive rate, and to reduce larger increases to an
inequality between two convolutions on the real line.

For $\lambda>0$, define
\begin{equation}
U_\lambda(\theta)
\coloneqq\frac{\sinh(\lambda/n)}{\cosh(\lambda/n)-\cos\theta},
\qquad
V_\lambda(\theta)
\coloneqq-\frac{\sin\theta}{\cosh(\lambda/n)-\cos\theta}.
\label{eq:UV-definitions}
\end{equation}
The denominator is positive, so $U_\lambda(\theta)>0$. These functions
arise by summing the geometric sequences in the solution family coefficients.
We suppress their dependence on $\theta$ when it is fixed.

\begin{lemma}[A sufficient condition for nonnegativity]
\label{lem:finite-hybrid}
For positive centers $\alpha,\beta$,
\begin{align}
q_{S_{\beta e^Y}+A_{\alpha e^Y}}(e^{i\theta})
={}&\mathbb E\!\left[\tanh(\beta e^Y/2)U_{\beta e^Y}\right]
\cos^2\!\left(\frac{n\theta}{2}\right)\notag\\
&+\mathbb E\!\left[\coth(\alpha e^Y/2)U_{\alpha e^Y}\right]
\sin^2\!\left(\frac{n\theta}{2}\right)\notag\\
&+\mathbb E\!\left[V_{\alpha e^Y}-V_{\beta e^Y}\right]
\sin\!\left(\frac{n\theta}{2}\right)
\cos\!\left(\frac{n\theta}{2}\right).
\label{eq:finite-hybrid-identity}
\end{align}
Consequently, this polynomial is nonnegative if, for every $\theta$,
\begin{equation}
\left|\mathbb E\!\left[V_{\alpha e^Y}-V_{\beta e^Y}\right]\right|^2
\le
4\mathbb E\!\left[\tanh(\beta e^Y/2)U_{\beta e^Y}\right]
\mathbb E\!\left[\coth(\alpha e^Y/2)U_{\alpha e^Y}\right].
\label{eq:finite-determinant-condition}
\end{equation}
Interchanging $\alpha$ and $\beta$ gives a sufficient condition for
$q_{S_{\alpha e^Y}+A_{\beta e^Y}}\ge0$.
\end{lemma}

\begin{proof}
Expanding the hyperbolic functions in \eqref{eq:green-parts} into
exponentials and summing the resulting geometric sequences gives
\begin{align*}
2\sum_{j=1}^{n-1}s_{\lambda,n}(j)\cos(j\theta)
={}&\bigl(\tanh(\lambda/2)U_\lambda-1\bigr)
\cos^2\!\left(\frac{n\theta}{2}\right)
-V_\lambda\sin\!\left(\frac{n\theta}{2}\right)
\cos\!\left(\frac{n\theta}{2}\right),\\
2\sum_{j=1}^{n-1}a_{\lambda,n}(j)\cos(j\theta)
={}&\bigl(\coth(\lambda/2)U_\lambda-1\bigr)
\sin^2\!\left(\frac{n\theta}{2}\right)
+V_\lambda\sin\!\left(\frac{n\theta}{2}\right)
\cos\!\left(\frac{n\theta}{2}\right).
\end{align*}
For example, each geometric sum is evaluated using
$\sum_{j=1}^{n-1}z^j=(z-z^n)/(1-z)$, with
$z=e^{\pm\lambda/n+i\theta}$. Adding the constant term and averaging
the two identities at their respective rates proves
\eqref{eq:finite-hybrid-identity}.

The right-hand side is a quadratic form in
$\cos(n\theta/2)$ and $\sin(n\theta/2)$. Its diagonal coefficients are
positive. Condition \eqref{eq:finite-determinant-condition} says that
its determinant is nonnegative, and therefore makes the quadratic
form nonnegative.
\end{proof}

To apply this condition, we express $U_\lambda$ and $V_\lambda$ as sums
of simpler functions on the real line:
\begin{align}
U_\lambda(\theta)
&=\sum_{k\in\mathbb Z}
\frac{2(\lambda/n)}{(\lambda/n)^2+(\theta+2\pi k)^2},
\label{eq:U-periodization}\\
V_\lambda(\theta)
&=-\lim_{K\to\infty}\sum_{k=-K}^{K}
\frac{2(\theta+2\pi k)}{(\lambda/n)^2+(\theta+2\pi k)^2}.
\label{eq:V-periodization}
\end{align}
Summing translates of a function in this way is called
\emph{periodization}: it produces a $2\pi$-periodic function of
$\theta$. To justify these identities, use the partial-fraction
expansion
\[
\coth z=\frac1z+\sum_{k=1}^{\infty}\frac{2z}{z^2+\pi^2k^2},
\]
which follows by logarithmically differentiating
$\sinh z=z\prod_{k\ge1}(1+z^2/(\pi^2k^2))$.
Taking real and imaginary parts at
$z=(\lambda/n+i\theta)/2$ gives
\eqref{eq:U-periodization}--\eqref{eq:V-periodization}.
The sum for $U_\lambda$ converges absolutely. Although the sum for
$V_\lambda$ is taken symmetrically, the difference at two positive
rates also converges absolutely: its summands are $O(|k|^{-3})$.

\begin{lemma}[Progress at every positive rate]
\label{lem:uniform-factor-three}
For every $1\le r\le3$, $n\ge2$, $\alpha>0$, and common probability
density $h$, both polynomials
\[
q_{S_{r\alpha e^Y}+A_{\alpha e^Y}},
\qquad
q_{S_{\alpha e^Y}+A_{r\alpha e^Y}}
\]
are nonnegative. In particular,
\begin{equation}
\Gamma_S(\alpha)\ge3\alpha,
\qquad
\Gamma_A(\alpha)\ge3\alpha.
\label{eq:uniform-greedy-progress}
\end{equation}
\end{lemma}

\begin{proof}
We first compare two individual rates $\lambda$ and $r\lambda$.
For a fixed summand in
\eqref{eq:U-periodization}--\eqref{eq:V-periodization}, put
$\tau=\theta+2\pi k$, $a=\lambda/n$, and $b=r\lambda/n$. Write
\[
u_a=\frac{2a}{a^2+\tau^2},
\qquad
v_a=\frac{2\tau}{a^2+\tau^2},
\]
and similarly for $b$. Then
\begin{equation}
\frac{(v_a-v_b)^2}{4u_au_b}
=\frac{\tau^2(b^2-a^2)^2}
{4ab(a^2+\tau^2)(b^2+\tau^2)}
\le\frac{(b-a)^2}{4ab}
=\frac{(r-1)^2}{4r}.
\label{eq:single-alias-ratio}
\end{equation}
The inequality follows from
$(a^2+\tau^2)(b^2+\tau^2)-\tau^2(a+b)^2=(\tau^2-ab)^2\ge0$.

For the polynomial whose antisymmetric component has rate $r\lambda$,
the product of the two weights in
\eqref{eq:finite-determinant-condition} is
\[
\tanh(\lambda/2)\coth(r\lambda/2)
=\frac{\tanh(\lambda/2)}{\tanh(r\lambda/2)}\ge\frac1r
\]
by concavity of $\tanh$ on $[0,\infty)$. For the other polynomial the
product is its reciprocal, which is at least one. Since
$(r-1)^2/(4r)\le1/r$ when $r\le3$,
\eqref{eq:single-alias-ratio} gives the required determinant bound
for each summand, with either assignment of the two weights.

The triangle inequality and Cauchy--Schwarz combine these bounds over
$k$: if $|b_k|\le2\sqrt{a_kc_k}$ with $a_k,c_k\ge0$, then
\[
\left|\sum_k b_k\right|
\le2\sum_k\sqrt{a_kc_k}
\le2\sqrt{\left(\sum_k a_k\right)\left(\sum_k c_k\right)}.
\]
Thus \eqref{eq:finite-determinant-condition} holds for both single-rate
polynomials. Averaging them over $\lambda=\alpha e^Y$ preserves
nonnegativity and proves the claim.
\end{proof}

For larger steps, averaging over rates becomes useful. Define
\begin{equation}
\mathfrak a_h(x)
\coloneqq\int_{\mathbb R}\frac{h(y)}{2\cosh(x-y)}\,dy,
\qquad
\mathfrak b_h(x)
\coloneqq\int_{\mathbb R}\frac{h(y)}{1+e^{2(y-x)}}\,dy.
\label{eq:continuum-ab}
\end{equation}
The next lemma reduces the required progress bound to a single
inequality involving these functions. Its hypothesis does not depend
on $n$.

\begin{lemma}[A sufficient
condition for progress]
\label{lem:alias-transfer}
Suppose that $d>0$, $0<\rho<1$, and
\begin{equation}
\mathfrak b_h(x)-\mathfrak b_h(x-d)
\le2\rho\sqrt{\mathfrak a_h(x)\mathfrak a_h(x-d)}
\qquad(x\in\mathbb R).
\label{eq:continuum-determinant}
\end{equation}
Let $\beta=e^d\alpha$. If all rates in both mixtures are at least $M>0$
and
\begin{equation}
\tanh(M/2)>\rho^2,
\label{eq:large-rate-slack}
\end{equation}
then both $q_{S_{\beta e^Y}+A_{\alpha e^Y}}$ and
$q_{S_{\alpha e^Y}+A_{\beta e^Y}}$ are nonnegative for every $n\ge2$.
\end{lemma}

\begin{proof}
Fix $\theta$ and one nonzero shifted angle $\tau=\theta+2\pi k$.
For this summand, write
\[
u_\lambda=\frac{2(\lambda/n)}{(\lambda/n)^2+\tau^2},
\qquad
v_\lambda=\frac{2\tau}{(\lambda/n)^2+\tau^2},
\qquad
x=\ln\frac{n|\tau|}{\alpha}.
\]
Substitution into \eqref{eq:continuum-ab} gives
\begin{align*}
\mathbb E[u_{\alpha e^Y}]
=\frac{2}{|\tau|}\mathfrak a_h(x),
\qquad
\mathbb E[u_{\beta e^Y}]
=\frac{2}{|\tau|}\mathfrak a_h(x-d),\\
\left|\mathbb E[v_{\alpha e^Y}-v_{\beta e^Y}]\right|
=\frac{2}{|\tau|}
\bigl(\mathfrak b_h(x)-\mathfrak b_h(x-d)\bigr).
\end{align*}
The last difference is nonnegative because $\mathfrak b_h$ is increasing.
The assumed 
inequality therefore implies
\[
\left|\mathbb E[v_{\alpha e^Y}-v_{\beta e^Y}]\right|
\le2\rho\sqrt{
\mathbb E[u_{\alpha e^Y}]\mathbb E[u_{\beta e^Y}]}.
\]
Since every rate is at least $M$,
\begin{align*}
\mathbb E[\tanh(\beta e^Y/2)u_{\beta e^Y}]
&\ge\tanh(M/2)\mathbb E[u_{\beta e^Y}],\\
\mathbb E[\coth(\alpha e^Y/2)u_{\alpha e^Y}]
&\ge\mathbb E[u_{\alpha e^Y}].
\end{align*}
Together with \eqref{eq:large-rate-slack}, these bounds give
\[
\left|\mathbb E[v_{\alpha e^Y}-v_{\beta e^Y}]\right|
\le2\sqrt{
\mathbb E[\tanh(\beta e^Y/2)u_{\beta e^Y}]
\mathbb E[\coth(\alpha e^Y/2)u_{\alpha e^Y}]}.
\]
The same argument works with $\alpha$ and $\beta$ interchanged on the
right. If $\tau=0$, both $v$ terms vanish, so the inequality holds
there as well.

Summing over $k$ by the triangle inequality and Cauchy--Schwarz, as in
the preceding proof, gives \eqref{eq:finite-determinant-condition}
for both assignments of the centers. Nonnegativity follows from \Cref{lem:finite-hybrid}.
\end{proof}

Although \eqref{eq:continuum-determinant} concerns functions on the
real line, the conclusion is nonnegativity of the finite Laurent
polynomials themselves. The identities
\eqref{eq:U-periodization}--\eqref{eq:V-periodization} introduce no
approximation.
The constant $\pi$ can already be seen in this sufficient condition.

\begin{lemma}
\label{lem:pi-speed-limit}
If \eqref{eq:continuum-determinant} holds with $\rho\le1$, then
$d\le\rho\pi\le\pi$.
\end{lemma}

\begin{proof}
Since $h$ has total mass $1$, integration gives
\[
\int_{\mathbb R}\mathfrak a_h(x)\,dx=\frac\pi2
\]
and
\[
 \int_{\mathbb R}[\mathfrak b_h(x)-\mathfrak b_h(x-d)]\,dx
 =\int_{\mathbb R} h(y)\int_0^d\int_{\mathbb R}\ell'(x-y-v)\,dx\,dv\,dy=d.
\]
where in the second identity, we used that the derivative of
$\ell(x)=(1+e^{-2x})^{-1}$ is $\operatorname{sech}^2(x)/2$, whose integral is
$1$; the integrals can be reordered by Tonelli's theorem. Thus Cauchy--Schwarz gives
\[
d\le2\rho\int_{\mathbb R}
\sqrt{\mathfrak a_h(x)\mathfrak a_h(x-d)}\,dx
\le2\rho\int_{\mathbb R}\mathfrak a_h(x)\,dx
=\rho\pi. \qedhere
\]
\end{proof}

If $h$ varies slowly over intervals of bounded length,
then one expects
\[
\mathfrak a_h(x)\approx\frac\pi2 h(x),
\qquad
\mathfrak b_h(x)-\mathfrak b_h(x-d)\approx d\,h(x).
\]
These approximations suggest that
\eqref{eq:continuum-determinant} should hold with
$\rho\approx d/\pi$, where $d=\ln(\beta/\alpha)$ is the
proposed increase in log-rate in one step. For each fixed
$0<d<\pi$, this suggests choosing a sufficiently broad density
to permit multiplicative progress by $e^d$ at sufficiently
large rates. Next we make this argument precise.

\subsection{Near-\texorpdfstring{$\pi$}{pi} compact mixtures}

We now choose the packet density
\begin{equation}
    h_W(y)
      \coloneqq Z_W^{-1}\cos^4\!\left(\frac{\pi y}{2W}\right)
         \mathbf 1_{\{|y|<W\}},
    \qquad
    Z_W\coloneqq\frac{3W}{4}.
    \label{eq:quartic-density}
\end{equation}
In the bulk of its support, $h_W$ varies only on scale~$W$ and is therefore
nearly constant on the unit scale of the kernels in
\eqref{eq:continuum-ab}.  If $h$ were exactly constant with value~$h_0$, then
formally
\[
    \mathfrak a_h(x)=\frac{\pi h_0}{2},
    \qquad
    \mathfrak b_h(x)-\mathfrak b_h(x-d)=dh_0,
\]
so \eqref{eq:continuum-determinant} would hold with
equality for $\rho=d/\pi$.
Thus a broad, nearly flat packet should permit every $d<\pi$.

The only obstruction is the boundary of the compact support.  The fourth
power in \eqref{eq:quartic-density} is chosen because, near an endpoint,
$h_W(W-s)$ is proportional to $s^4$.  The following lemma verifies that this
quartic edge retains the same factor $d/\pi$.

Put
\begin{equation}
    \ell(u)\coloneqq\frac{1+\tanh u}{2},
    \qquad
    L_d(u)\coloneqq\ell(u)-\ell(u-d),
    \label{eq:logistic-increment}
\end{equation}
and, for $r\ge0$, define
\begin{equation}
    \mathcal A_r(x)
       \coloneqq\int_0^\infty s^r\operatorname{sech}(x+s)\,ds,
    \qquad
    \mathcal D_{4,d}(x)
       \coloneqq\int_0^\infty s^4L_d(x+s)\,ds.
    \label{eq:quartic-edge-integrals}
\end{equation}

\begin{lemma}[The inequality at a quartic endpoint]
\label{lem:quartic-edge}
There is a constant $0<d_0<\pi$ such that, for
$d_0\le d\le\pi$ and every $x\in\mathbb R$,
\begin{equation}
\mathcal D_{4,d}(x)
\le\frac d\pi
\sqrt{\mathcal A_4(x)\mathcal A_4(x-d)}.
\label{eq:quartic-edge-bound}
\end{equation}
In particular, the determinant ratio
(the left-hand side of \eqref{eq:continuum-determinant}
divided by its right-hand side with $\rho=1$) 
is strictly less than $1$ in this model whenever $d<\pi$.
\end{lemma}

\begin{proof}
The proof compares $L_d$ with a combination of $\operatorname{sech}$
and its derivatives. We choose this combination to have the same
first four moments as $L_d$. A sign calculation then shows that it
has a larger integral against $(u-x)_+^4$, where
$(x)_+ \coloneqq \max\{x,0\}$. Integration by parts and Cauchy--Schwarz
give the desired bound.

\emph{Matching the moments.}
Let $f(u)=\operatorname{sech}u$ and put
\begin{equation}
G_d\coloneqq\frac d\pi\left[
f-\frac d2 f'
+\frac{2d^2-\pi^2}{12}f''
+\frac{d(\pi^2-d^2)}{24}f'''
\right],
\qquad H_d \coloneqq G_d-L_d.
\label{eq:Gd-definition}
\end{equation}
The coefficients are chosen so that
\begin{equation}
\int_{\mathbb R}u^jH_d(u)\,du=0
\qquad(0\le j\le3).
\label{eq:four-vanishing-moments}
\end{equation}
Explicitly, the moments of $f$ through degree
three are $\pi,0,\pi^3/4,0$. Also,
$L_d(u)=\int_0^d\ell'(u-v)\,dv$, where
$\ell'(u)=\operatorname{sech}^2(u)/2$ has mass one, mean zero, and
second moment $\pi^2/12$. It follows that the corresponding moments
of $L_d$ are
\[
d,\qquad \frac{d^2}{2},\qquad
\frac{d^3}{3}+\frac{d\pi^2}{12},\qquad
\frac{d^4}{4}+\frac{d^2\pi^2}{8}.
\]
Integration by parts in \eqref{eq:Gd-definition} gives precisely
these four values for $G_d$, proving
\eqref{eq:four-vanishing-moments}.

\emph{Determining the signs.}
We first take $d=\pi$, for which
$G_\pi=f-(\pi/2)f'+(\pi^2/12)f''$.
With $z=e^u>0$, multiplication by a positive denominator shows that
$H_\pi(u)$ has the sign of
\begin{align}
P(z)&\coloneqq\left[
2-\pi+\frac{\pi^2}{6}
+(4-\pi^2)z^2
+\left(2+\pi+\frac{\pi^2}{6}\right)z^4
\right](1+e^{-2\pi}z^2)\notag\\
&\hspace{2em}-(1-e^{-2\pi})z(1+z^2)^2.
\label{eq:quartic-sign-polynomial}
\end{align}
The nonzero coefficient signs, in descending order, are
$+,-,+,-,-,-,+$, so Descartes' rule gives at most four positive
zeros, counted with multiplicity. Direct substitution gives the signs
\[
\begin{array}{c|ccccc}
z&1/10&1/2&2&10&100\\ \hline
\operatorname{sgn}P(z)&+&-&+&-&+
\end{array}
\]
and hence there are exactly four positive zeros, all simple.
Therefore $H_\pi$ has exactly four sign changes and is positive in
both tails.

The same pattern holds for $d$ sufficiently close to $\pi$. To see
this, take disjoint small intervals around the four zeros. Simplicity
of the zeros preserves one crossing in each interval under a small
change in $d$. On any remaining compact set, $H_\pi$ is bounded away
from zero, so its sign is unchanged. Finally, the tails remain
positive uniformly for $d$ near $\pi$: $L_d(u)=O(e^{-2|u|})$, whereas
$G_d(u)/e^u$ as $u\to-\infty$ and $G_d(u)/e^{-u}$ as $u\to\infty$
have limits continuous in $d$ and equal at $d=\pi$ to
\[
2-\pi+\frac{\pi^2}{6}>0,
\qquad
2+\pi+\frac{\pi^2}{6}>0,
\]
respectively. The tail expansions are uniform on a fixed neighborhood
of $\pi$. Choosing $d_0<\pi$ sufficiently close to $\pi$ therefore
gives exactly four simple sign changes, with positive tails, for
every $d\in[d_0,\pi]$.

\emph{Using the moments and signs.}
Fix such a $d$, and let $u_1<u_2<u_3<u_4$ be the four zeros of $H_d$.
For a smooth function $\varphi$ with $\varphi^{(4)}\ge0$, let $p$ be
its cubic interpolating polynomial at these zeros. The interpolation
remainder has the sign of
$\prod_{i=1}^4(u-u_i)$, or is zero. This is also the sign of $H_d(u)$.
Thus $(\varphi-p)H_d\ge0$, and the vanishing moments imply
\[
\int_{\mathbb R}\varphi(u)H_d(u)\,du
=\int_{\mathbb R}(\varphi(u)-p(u))H_d(u)\,du\ge0
\]
whenever these integrals converge.
The same conclusion holds for $\varphi(u)=(u-x)_+^4$, by smoothing
this function with a nonnegative smooth kernel and passing to the
limit. Its fourth derivative is nonnegative, and its polynomial
growth is integrable against the exponentially decaying $H_d$.

Applying the comparison to this function and integrating by parts
gives
\begin{align}
\mathcal D_{4,d}(x)
&\le\frac d\pi\left[
\mathcal A_4(x)+2d\mathcal A_3(x)
+(2d^2-\pi^2)\mathcal A_2(x)
-d(\pi^2-d^2)\mathcal A_1(x)
\right]\notag\\
&\le\frac d\pi\left[
\mathcal A_4(x)+2d\mathcal A_3(x)+d^2\mathcal A_2(x)
\right]\notag\\
&=\frac d\pi\int_0^\infty
s^2(s+d)^2\operatorname{sech}(x+s)\,ds.
\label{eq:D-B-bound}
\end{align}
For the second inequality, use $d\le\pi$ and
$\mathcal A_1,\mathcal A_2\ge0$. Finally, Cauchy--Schwarz gives
\begin{align*}
\int_0^\infty s^2(s+d)^2\operatorname{sech}(x+s)\,ds
&\le\sqrt{
\mathcal A_4(x)
\int_0^\infty(s+d)^4\operatorname{sech}(x+s)\,ds}\\
&\le\sqrt{\mathcal A_4(x)\mathcal A_4(x-d)}.
\end{align*}
The last inequality follows by setting $t=s+d$ and extending the
resulting integral from $t\ge d$ to $t\ge0$.
\end{proof}

We can now establish the condition of \cref{lem:alias-transfer} for the density $h_W$.

\begin{proposition}[Near-$\pi$ compact mixtures]
\label{prop:compact-mixtures}
There exist sequences $W_m\to\infty$, $d_m\uparrow\pi$, and $\rho_m<1$ such
that the density $h_{W_m}$ satisfies
\begin{equation}
    |\mathfrak b_{h_{W_m}}(x)-\mathfrak b_{h_{W_m}}(x-d_m)|
       \le2\rho_m
          \sqrt{
             \mathfrak a_{h_{W_m}}(x)
             \mathfrak a_{h_{W_m}}(x-d_m)}
    \label{eq:compact-mixture-certificate}
\end{equation}
for every $x\in\mathbb R$.
\end{proposition}

\begin{proof}
Write $k(u)\coloneqq\operatorname{sech}(u)/2$.  Then
\[
    \mathfrak a_h=h*k,
    \qquad
    \mathfrak b_h(x)-\mathfrak b_h(x-d)=h*L_d(x).
\]
Fix any $d<\pi$ sufficiently close to~$\pi$.  We claim that the strict form of
\eqref{eq:continuum-determinant} holds for $h_W$ once $W$ is sufficiently
large.  Suppose otherwise and choose $W_j\to\infty$ and points $x_j$ where
the determinant ratio is at least~$1$.  Put $t_j=x_j-d/2$.  There are three
possible limiting regimes.

First suppose $t_j$ lies inside the support and its distances from both
endpoints tend to infinity.  On every fixed window,
\[
    \frac{h_{W_j}(t_j+u)}{h_{W_j}(t_j)}\longrightarrow1.
\]
The kernels decay exponentially, while concavity of the sine implies a
polynomial upper bound for the ratio. Dominated convergence therefore
shows that the determinant ratio tends to $d/\pi<1$, as predicted by the
constant-density calculation.

Second, suppose the signed distance from an endpoint stays bounded.  Since
\begin{equation}
    W^5h_W(W-s)\longrightarrow\frac{\pi^4}{12}s_+^4,
    \qquad
    W^5h_W(W-s)\le\frac{\pi^4}{12}s_+^4,
    \label{eq:quartic-edge-limit}
\end{equation}
dominated convergence reduces the determinant ratio to the quartic-edge model
of \cref{lem:quartic-edge}, where it is strictly less than~$1$.  The left
endpoint follows by reflection.

Finally, suppose $t_j$ lies outside the right endpoint at a distance tending
to infinity.  If
\[
    I_p=\int h_W(y)e^{py}\,dy,
\]
the kernel tails give
\[
    \mathfrak a_{h_W}(x)= \Theta(e^{-x}I_1),
    \qquad
    h_W*L_d(x)=O(e^{-2x}I_2).
\]
Since $I_2\le e^WI_1$, the determinant ratio is exponentially small in the
distance beyond the endpoint.  The left tail is symmetric.

These regimes exhaust all subsequences, contradicting a ratio at least~$1$.
Thus, for each fixed $d<\pi$ sufficiently close to~$\pi$, all sufficiently
large $W$ admit a uniform constant $\rho<1$.  Applying this successively to a
sequence $d_m\uparrow\pi$ gives the claim.
\end{proof}

\subsection{The endpoints and the greedy recurrence}

\Cref{lem:uniform-factor-three} guarantees progress at every positive center, while \cref{lem:alias-transfer} gives the stronger progress bound at sufficiently large centers.  The initial
center $\alpha=0$ is degenerate and is handled directly.

\begin{lemma}[First move]
\label{lem:first-move}
If $0\le\beta e^Y\le1$ almost surely, then
$q_{S_{\beta e^Y}+A_0}\ge0$. In particular, for $h=h_W$,
\[
\Gamma_S(0)\ge e^{-W}.
\]
\end{lemma}

\begin{proof}
For a single rate, positivity reduces to
\begin{equation}
    4\tanh(\lambda/2)
       \ge n\tanh(\lambda/(2n)).
    \label{eq:first-move-inequality}
\end{equation}
For $0\le\lambda\le1$, the right-hand side is at most $\lambda/2$, while the
left-hand side is larger.  The mixture is an average of these nonnegative
single-rate hybrids and is therefore nonnegative.
\end{proof}

At the other endpoint, it is unnecessary to send the solution family rates literally
to infinity.  Once the remaining coefficients have sufficiently small
$\ell_1$-mass, two parity-respecting moves remove them exactly.

\begin{lemma}[Two-query exactifier]
\label{lem:two-query-exactifier}
Suppose a feasible path ends at $R(z)=1+\sum_{j=1}^{n-1}a_j(z^j+z^{-j})$
and
\begin{equation}
    2\sum_{j=1}^{n-1}|a_j|<1.
    \label{eq:exactifier-l1}
\end{equation}
Then two additional queries terminate exactly at the constant polynomial~$1$.
\end{lemma}

\begin{proof}
Set the component permitted by the next query to zero, leaving the
other component unchanged. Since $J$ permutes coordinates,
\[
\|P_\pm a\|_1
=\left\|\frac{a\pm Ja}{2}\right\|_1\le\|a\|_1.
\]
The resulting polynomial is therefore at least
$1-2\|a\|_1>0$ everywhere on the unit circle. The following query
sets its remaining component to zero. Both increments have the
required reversal parity.
\end{proof}

\begin{lemma}[Large-rate decay]
\label{lem:large-rate-decay}
If $\lambda\ge Dn>0$, then
\begin{equation}
\sum_{j=1}^{n-1}g_{\lambda,n}(j)
\le\frac1{(e^D-1)(1-e^{-2Dn})}.
\label{eq:large-rate-sum}
\end{equation}
In particular, $D=2$ suffices for the following conclusion: if
$\alpha e^Y\ge Dn$ and $\beta e^Y\ge Dn$ almost surely, the
coefficients of $q_{S_{\alpha e^Y}+A_{\beta e^Y}}$ satisfy
\eqref{eq:exactifier-l1}.
\end{lemma}

\begin{proof}
The identity
\[
g_{\lambda,n}(j)
=\frac{e^{-\lambda j/n}-e^{-2\lambda+\lambda j/n}}
{1-e^{-2\lambda}}
\le\frac{e^{-Dj}}{1-e^{-2Dn}}
\]
and a geometric-series sum prove \eqref{eq:large-rate-sum}.
Because $g_{\lambda,n}(j)\ge0$,
\[
\sum_j|s_{\lambda,n}(j)|=\sum_j g_{\lambda,n}(j),
\qquad
\sum_j|a_{\lambda,n}(j)|\le\sum_j g_{\lambda,n}(j).
\]
Taking expectations and using the triangle inequality gives
\[
2\sum_{j=1}^{n-1}
\left|S_{\alpha e^Y}(j)+A_{\beta e^Y}(j)\right|
\le\frac4{(e^D-1)(1-e^{-2Dn})}.
\]
For $D=2$ and $n\ge2$, the last expression is less than $1$.
\end{proof}

We now bound the number of steps used by the greedy recurrence using the preceding lower bounds.  This is
the only place where different physical-rate ranges enter the argument.

\begin{lemma}[Progress of the greedy recurrence]
\label{lem:greedy-recurrence}
Fix $W,d,\rho,M$, and let $h=h_W$.  Suppose
\eqref{eq:continuum-determinant} holds and
\[
    \tanh(M/2)>\rho^2.
\]
Let $D$ be the sufficiently large fixed constant furnished by
\cref{lem:large-rate-decay}, and assume that $Dn\ge M \ge 1$.
Let $\alpha_t$ be the greedy recurrence
\eqref{eq:greedy-scale-recurrence}.  The path reaches a hybrid satisfying
\eqref{eq:exactifier-l1}, and hence reaches $q=1$, using at most
\begin{equation}
    1+
    \left\lceil\frac{2W+\ln M}{\ln3}\right\rceil
    +
    \left\lceil\frac{\ln(Dn/M)}{d}\right\rceil
    +3
    \label{eq:greedy-query-bound}
\end{equation}
queries.
\end{lemma}

\begin{proof}
Because $Y\in[-W,W]$, the packet centered at $e^{-W}$ has rates at
most~$1$.  \cref{lem:first-move} therefore gives
\begin{equation}
    \alpha_1=\Gamma_S(0)\ge e^{-W}.
    \label{eq:greedy-first-lower-bound}
\end{equation}
For every later finite scale, \cref{lem:uniform-factor-three} gives
\[
    \alpha_t\ge3\alpha_{t-1}.
\]
We have
\begin{equation}
    t_M \coloneqq \min\{t:\alpha_t e^{-W}\ge M\}
      \le
      1+\left\lceil\frac{2W+\ln M}{\ln3}\right\rceil.
    \label{eq:greedy-threshold-time}
\end{equation}

Once $\alpha_t e^{-W}\ge M$, every rate in the packet centered at
$\alpha_t$ is at least~$M$.  The assumed continuum certificate and
\cref{lem:alias-transfer} show that both greedy transition maps admit the
candidate $e^d\alpha_t$.  Hence
\begin{equation}
    \alpha_{t+1}\ge e^d\alpha_t
    \qquad(t\ge t_M).
    \label{eq:greedy-high-rate-recurrence}
\end{equation}
After at most
\[
    \left\lceil\frac{\ln(Dn/M)}d\right\rceil
\]
such updates, one of the scales is at least $Dn e^W$.  One further greedy
update makes this the smaller of the two component centers in the current
hybrid.  Both physical-rate supports are then at least $Dn$, so
\cref{lem:large-rate-decay} gives \eqref{eq:exactifier-l1}.
\cref{lem:two-query-exactifier} appends two final queries.  Combining
the bounds proves \eqref{eq:greedy-query-bound}.  If a greedy transition
reaches $\infty$ earlier, the path terminates sooner.
\end{proof}

\subsection{Choosing parameters}

Choose sequences $W_m,d_m,\rho_m$ from
\cref{prop:compact-mixtures}.  Choose $M_m\ge1$ so that
\begin{equation}
    \tanh(M_m/2)>\rho_m^2.
    \label{eq:choose-M}
\end{equation}
Finally, choose increasing thresholds $N_m\to\infty$ sufficiently rapidly
that
\[
    DN_m\ge M_m
\]
and
\begin{equation}
    \frac{W_m+\ln M_m+1}{\ln n}\le\frac1m
    \qquad(n\ge N_m).
    \label{eq:diagonal-threshold}
\end{equation}
Given $n$, take the largest $m=m(n)$ with $N_m\le n$.  Fix
$h=h_{W_{m(n)}}$ for the entire greedy path and use the transition maps in
\eqref{eq:greedy-S-map} and \eqref{eq:greedy-A-map}.

\begin{theorem}[Exact ordered search at the adversary constant]
\label{thm:ordered-search}
For every sufficiently large $n$, the greedy family
\eqref{eq:greedy-polynomial-path}, followed by the two-query exactifier,
gives nonnegative Laurent polynomials $q_0,\ldots,q_k$ satisfying
\eqref{eq:path-endpoints} and \eqref{eq:alternating-parity}, with
\begin{equation}
    k=\frac{\ln n}{\pi}+o(\log n).
    \label{eq:optimal-query-count}
\end{equation}
Consequently, exact quantum ordered search is possible using
\[
    \frac{1}{\pi}\ln n+o(\log n)
\]
queries.
\end{theorem}

\begin{proof}
The initial polynomial is $q_0=F_n$.  Every greedy hybrid is nonnegative by
definition, and the alternating choice of transition maps gives
\eqref{eq:alternating-parity}.  \cref{lem:greedy-recurrence} and the
choice of parameters give
\begin{align*}
    k(n)
      &\le
       \frac{\ln n}{d_{m(n)}}
       +O\!\left(
          W_{m(n)}+\ln M_{m(n)}+1
        \right)\\
      &=\frac{\ln n}{\pi}+o(\log n),
\end{align*}
because $d_{m(n)}\to\pi$ and
\eqref{eq:diagonal-threshold} makes the remaining terms
$o(\log n)$.  \cref{lem:large-rate-decay,lem:two-query-exactifier} give $q_k=1$ exactly.  \cref{prop:ordered-search} then produces the corresponding exact quantum query
algorithm.
\end{proof}

\section*{Acknowledgments}

This work received support from the National Science Foundation (grant 26-17356) and the Department of Energy (grant DE-SC0020264 and the Office of Science, Office of Advanced Scientific Computing Research, Accelerated Research in Quantum Computing program).

\printbibliography
\end{document}